%% file: scatteredUniversalityArxiv.tex
\documentclass[a4paper,11pt]{article}

\usepackage{microtype}

\usepackage{amsmath,amssymb}
\usepackage[utf8]{inputenc}
\usepackage{todonotes}
\usepackage{tikz}
\usepackage{mathtools}
\usepackage[linesnumbered,ruled,vlined]{algorithm2e}
\usepackage{appendix}

\usetikzlibrary{positioning}
\usepgflibrary{arrows}
\usetikzlibrary{arrows,decorations.pathreplacing}
\usepackage[noadjust]{cite}
\usepackage{thmtools}
\usepackage{thm-restate}
\usepackage{thm-restate}
\usepackage{amsthm}
\usepackage{amssymb}
\usepackage{amsmath}
\usepackage{amsfonts}
\usepackage{fullpage}
\usepackage{authblk}
\usepackage{url}
\usepackage{hyperref}
\usepackage{cleveref}

\title{The Edit Distance to $k$-Subsequence Universality} %TODO Please add

\author[1]{Pamela Fleischmann}
\author[2]{Maria Kosche}
\author[2]{Tore Ko\ss }
\author[2]{Florin Manea}
\author[2]{Stefan Siemer}

\affil[1]{Kiel University, Computer Science Department, Germany\\ \texttt{fpa@informatik.uni-kiel.de}}

\affil[2]{G\"ottingen University, Computer Science Department, Germany\\
	\texttt{\{maria.kosche, tore.koss, florin.manea, stefan.siemer\}@cs.uni-goettingen.de}}
\newtheorem{theorem}{Theorem}

\theoremstyle{definition}
\newtheorem{definition}{Definition}

\theoremstyle{remark}
\newtheorem{remark}{Remark}

\theoremstyle{lemma}
\newtheorem{lemma}{Lemma}

\theoremstyle{example}

\def\ta{\mathtt{a}}
\def\tb{\mathtt{b}}
\def\tc{\mathtt{c}}
\def\td{\mathtt{d}}

\newcommand{\letters}{\text{alph}}

\def\nth#1{#1$^{\text{th}}$}

\def\N{\mathbb{N}}

\DeclareMathOperator{\find}{\mathtt{find}}
\DeclareMathOperator{\union}{\mathtt{union}}

\DeclareMathOperator{\ScatFact}{Subseq}

\DeclareMathOperator{\ar}{ar}

\DeclareMathOperator{\last}{last}

\DeclareMathOperator{\pos}{pos}

\DeclareMathOperator{\RMQ}{RMQ}
\DeclareMathOperator{\argmin}{arg\ min}
\DeclareMathOperator{\univ}{univ}
\DeclareMathOperator{\freq}{freq}

\newcommand\tn{\mathtt{n}}
\newcommand*{\qqed}{\null\nobreak\hfill\ensuremath{\square}}%

\tikzset{ns/.style={
        inner sep = 0mm, outer sep=0mm,
        append after command={
            [very thick]
            (\tikzlastnode.north west)edge(\tikzlastnode.north east)
            [very thick]
            (\tikzlastnode.south west)edge(\tikzlastnode.south east)
        }
    }
}

\tikzset{wens/.style={
        inner sep = 0mm, outer sep=0mm,
        append after command={
            [very thick]
            (\tikzlastnode.north west)edge(\tikzlastnode.south west)
            [very thick]
            (\tikzlastnode.north east)edge(\tikzlastnode.south east)
            [very thick]
            (\tikzlastnode.north west)edge(\tikzlastnode.north east)
            [very thick]
            (\tikzlastnode.south west)edge(\tikzlastnode.south east)
        }
    }
}

\tikzset{wns/.style={
        inner sep = 0mm, outer sep=0mm,
        append after command={
            [very thick]
            (\tikzlastnode.north west)edge(\tikzlastnode.south west)
            [very thick]
            (\tikzlastnode.north west)edge(\tikzlastnode.north east)
            [very thick]
            (\tikzlastnode.south west)edge(\tikzlastnode.south east)
        }
    }
}

\tikzset{ens/.style={
        inner sep = 0mm, outer sep=0mm,
        append after command={
            [very thick]
            (\tikzlastnode.north east)edge(\tikzlastnode.south east)
            [very thick]
            (\tikzlastnode.north west)edge(\tikzlastnode.north east)
            [very thick]
            (\tikzlastnode.south west)edge(\tikzlastnode.south east)
        }
    }
}

\tikzset{e/.style={
        inner sep = 0mm, outer sep=0mm,
        append after command={
            [very thick]
            (\tikzlastnode.north east)edge(\tikzlastnode.south east)
            [very thick]
        }
    }
}

\tikzset{w/.style={
        inner sep = 0mm, outer sep=0mm,
        append after command={
            [very thick]
            (\tikzlastnode.north west)edge(\tikzlastnode.south west)
            [very thick]
        }
    }
}

\begin{document}
\maketitle

\begin{abstract}
A word $u$ is a subsequence of another word $w$ if $u$ can be obtained from $w$ by deleting some of its letters. The word $w$ with $\letters(w)=\Sigma$ is called $k$-subsequence universal if the set of subsequences of length $k$ of $w$ contains all possible words of length $k$ over $\Sigma$. We propose a series of efficient algorithms computing the minimal number of edit operations (insertion, deletion, substitution) one needs to apply to a given word in order to reach the set of $k$-subsequence universal words. 
\end{abstract}

\section{Introduction}
%%%%%%%%%%%%%%%%%%%%%%%
A word $v$ is a subsequence (also called scattered factor or subword) of a word $w$ if there exist (possibly empty) words 
 $x_1, \ldots, x_{\ell+1}$ and $v_1, \ldots, 
v_\ell$ such that $v = v_1 \ldots v_\ell$ and $w = x_1 v_1 \ldots x_\ell v_\ell x_{\ell+1}$. That is, $v$ is obtained from $w$ by removing some of its letters. 

The study of the relationship between words and their subsequences is a central topic in combinatorics on words and string algorithms, as well as in language and automata theory (see, e.g., 
the chapter {\em Subwords} by J. Sakarovitch and I. Simon in  \cite[Chapter 6]{Loth97} for an overview of the fundamental aspects of this topic).
The concept of subsequence and its generalisations play an important role in various areas of theoretical computer science. For instance, in logic of automata theory, subsequences are used in the context of piecewise testability~\cite{simonPhD,Simon72}, in particular to the height of piecewise testable languages~\cite{KarandikarKS15,CSLKarandikarS,journals/lmcs/KarandikarS19}, subword order \cite{HalfonSZ17,KuskeZ19,Kuske20}, or downward closures \cite{Zetzsche16}. In combinatorics on words, many concepts were developed around the idea of counting the occurrences of particular subsequences of a word, such as the $k$-binomial equivalence ~\cite{RigoS15,FreydenbergerGK15,LeroyRS17a,Rigo19}, subword histories ~\cite{Seki12}, and Parikh matrices ~\cite{Mat04,Salomaa05}. In the area of algorithms, subsequences appear, e.g., in classical problems such as {the longest common subsequence} \cite{DBLP:journals/tcs/Baeza-Yates91,DBLP:conf/fsttcs/BringmannC18,BringmannK18}, {the shortest common supersequence} \cite{Maier:1978}, or {the string-to-string correction} \cite{Wagner:1974}. 
From a practical point of view, subsequences are useful in bioinformatics-related scenarios, as well as in other areas where they model corrupted or lossy representations of an original string, see~\cite{sankoff}.

A major area of research related to subsequences is the study of the set of all subsequences of bounded length of a word, initiated by Simon in his PhD thesis \cite{simonPhD}. In particular, Simon defined and studied (see \cite{Simon72,Loth97}) the relation $\sim_k$ (called now Simon's congruence) between words having exactly the same set of subsequences of length at most $k$. The surveys \cite{Pin2004,Pin2019} overview some of the extensions of Simon's seminal work from 1972 in various areas related to automata theory. Moreover, $\sim_k$ is a well-studied relation in the area of string algorithms too. The problems of deciding whether two given words are $\sim_k$-equivalent, for a given $k$, and to find the largest $k$ such that two given words are $\sim_k$-equivalent (and their applications) were heavily investigated in the literature, see, e.g., \cite{TCS::Hebrard1991,garelCPM,SimonWords,DBLP:conf/wia/Tronicek02,DBLP:journals/jda/CrochemoreMT03,KufMFCS} and the references therein. This year, optimal solutions were given for both these problems \cite{DBLP:conf/dlt/BarkerFHMN20,mfcs2020}. In \cite{DBLP:conf/dlt/BarkerFHMN20} it was shown how to compute the shortlex normal form of a given word in linear time, i.e., the minimum representative of a $\sim_k$-equivalence class w.r.t. shortlex ordering. This can be directly applied to test whether two words are $\sim_k$-equivalent: they need to have the same shortlex normal form. In \cite{mfcs2020}, a tree-like structure, called Simon-tree, was used to represent the equivalence classes induced by $\sim_k$ on the set of suffixes of a word, for all possible values of $k$, and then, given two words, a correspondence between their Simon-trees was constructed to compute in linear time the largest $k$ for which they are $\sim_k$-equivalent.

Extending the algorithmic work on $\sim_k$, the following problem seems interesting: given two words $w$ and $u$ and an integer $k$, which is the minimal number of edit operations we need to perform on $w$ to obtain a word $v$ such that $v\sim_k u$? As edit operations we consider, as usual, insertion, deletion, and substitution of letters. Rephrasing, we ask how far (w.r.t. the {\em edit distance}) are two words from being $\sim_k$-equivalent, or, how far is the word $w$ from the set of all words which are $\sim_k$-equivalent to $u$. To this end, we can replace the target-word $u$ by a set $U$ of words of length $k$ and ask for the minimal number of edit operations we need to perform on $w$ to obtain a word $v$ whose set of subsequences of length $k$ is (or includes)~$U$.

This direction of research is not new, and has always been a source of interesting problems. One of the most classical and well-understood string-problems is computing the edit distance between words \cite{EditDistance}, for which an optimal (up to poly-logarithmic factors) solution exists~\cite{DBLP:journals/siamcomp/BackursI18,DBLP:journals/jcss/MasekP80}. The problem of computing the edit distance between a word and a language, or between two languages, is also a well-studied problem, in various settings (see, e.g., \cite{DBLP:conf/focs/BringmannGSW16,DBLP:conf/icalp/JayaramS17,DBLP:journals/ijfcs/HanKS13,DBLP:conf/dlt/CheonH20,DBLP:conf/dlt/CheonHKS19}).

In this paper, we make some initial steps in the study of the problems introduced above. While we do not solve their general form, we investigate one of their particular cases which seems interesting, meaningful, and well motivated. We follow the line of research of \cite{KarandikarKS15,CSLKarandikarS,dlt2019,DBLP:conf/dlt/BarkerFHMN20} and focus on a special $\sim_k$-class of words. A word $w$ is {\em $k$-subsequence universal} (for short $k$-universal) with respect to an alphabet $\Sigma$ if its set of subsequences of length $k$ equals $\Sigma^k$. In the problems we consider, $\Sigma$ will be the set $\letters(w)$ of letters occurring in the input $w$ (to this end, see the discussion in \autoref{alphabet}). The maximum $k$ for which a word $w$ is $k$-universal is {\em the universality index} of $w$. In this context, we consider the problem of computing for a given word $w$ and an integer $k$ the minimal number of edit operations we need to perform on $w$ in order to obtain a $k$-universal word. That is, we are interested in the edit distance from the input word $w$ to the set of $k$-universal words w.r.t. $\letters(w)$. 

Before presenting our results, we briefly discuss the motivation of considering $k$-subsequence universal words. Firstly, using the name {\em universal} in this context is not unusual. The classical universality problem (see, e.g., \cite{HolzerK11}) is whether a given language $L$  (over an alphabet $\Sigma$, and specified by an automaton or grammar) is equal to $\Sigma^{\ast}$. The works \cite{Rampersad:2012,KrotzschMT17,GawrychowskiRSS17} and the references therein discuss many variants of and results on the universality problem for various language generating and accepting formalisms. The universality problem was considered for words \cite{martin1934,Bruijn46} and  partial words \cite{ChenKMS17,GoecknerGHKKKS18} w.r.t. their factors. More precisely, one is interested in finding, for a given $\ell$, a word $w$ over an alphabet $\Sigma$, such that each word of length $\ell$ over $\Sigma$ occurs exactly once as a contiguous factor of $w$. De Bruijn sequences~\cite{Bruijn46} fulfil this property and have many applications in computer science or combinatorics, see \cite{ChenKMS17,GoecknerGHKKKS18} and the references therein. In~\cite{KarandikarKS15,CSLKarandikarS,journals/lmcs/KarandikarS19} the authors define the notion of $k$-rich words in relation to the height of piecewise testable languages, a class of simple regular languages with applications in learning theory, databases theory, or linguistics (see \cite{journals/lmcs/KarandikarS19} and the references therein). The class of $k$-rich words coincides with that of $k$-subsequence universal words. The study of $k$-subsequence universal words was continued, from a combinatorial point of view, in \cite{dlt2019,DBLP:conf/dlt/BarkerFHMN20}. So, it seems that investigating this class also from an algorithmic perspective is motivated by, fits in, and even enriches this well-developed and classical line of research. 

{\bf Our results.} Firstly, we note that when we want to increase the universality of a word by edit operations, it is enough to use only insertions. Similarly, when we want to decrease the universality of a word, it is enough to consider deletions. So, to measure the edit distance to the class of $k$-subsequence universal words, for a given $k$, it is enough to consider either insertions or deletions. However, changing the universality of a word by substitutions (both increasing and decreasing it) is interesting in itself as one can see the minimal number of substitutions needed to transform a word $w$ into a $k$-universal word as the {\em Hamming distance}~\cite{6772729} between $w$ and the set of $k$-universal words. Thus, we consider all these operations independently and propose efficient algorithms computing the minimal number of insertions, deletions, and substitutions, respectively, needed to apply to a given word $w$ in order to reach the class of $k$-universal words (w.r.t. the alphabet of $w$), for a given $k$. The time needed to compute these numbers is $O(nk)$ in the case of deletions and substitutions, as well as in the case of insertions if $k\in O(c^n)$ for some constant $c$ (for even larger values of $k$ we need to add roughly the time complexity of multiplying $n$ and~$k$). The respective algorithms are presented in \autoref{editDist}. 

Our algorithms are based, like most edit distance algorithms, on a dynamic programming approach. However, implementing such an approach within the time complexities stated above does not seem to follow directly from the known results on the word-to-word or word-to-language edit distance. In particular, we do not explicitly construct any $k$-universal word nor any representation (e.g., automaton or grammar) of the set of $k$-universal words, when computing the distance from the input word $w$ to this set. Rather, we can compute the $k$-universal word which is closest w.r.t. edit distance to $w$ as a byproduct of~our~algorithms. In our approach, we first develop (\autoref{toolbox}) several efficient data structures (most notably \autoref{speedUp}). Then (\autoref{editDist}), for each of the considered operations, we make several combinatorial observations, allowing us to restrict the search space of our algorithms, and creating a framework where our data~structures can be used efficiently. 
%For space reasons, some proofs, examples, and pseudocode for the algorithms are given~in~the~Appendix.

\input{prels}

%%%%%%%%%%%%%%%%%%%%%%%%%%%%

\section{Toolbox}
In this section, we present data structures which will be decisive in obtaining efficient solutions for the approached problems. Our running example for \autoref{toolbox} will be the word $w=\tb \ta \tn \ta \tn \ta \tb \ta \tn $, on which we illustrate some of the notions we define here. 
Detailed examples are given in \autoref{examples}.

\subsection{The Tools}
\label{toolbox}

For a word $w$ over an alphabet $\Sigma$, a position $j$ of $w$, and a letter $a\in \Sigma$ which occurs in $w[1:j]$, let $\last_j[a]=\max\{i\leq j\mid w[i]=a\}$, the last position where $a$ occurs before $j$; if $a$ does not occur in $w[1:j]$ or for $j=0$, then, by convention, $\last_j[a]=|w|+1$. Let $S_j=\{\last_j[a]\mid a\in \letters(w[1:j])\}$. If $i,j$ are two positions of $w$, let $\Delta(i,j)$ be the number of distinct letters occurring in $w[i:j]$, i.e., $\Delta(i,j)=|\letters(w[i:j])|$; if $i>j$, then $\Delta(i,j)=0$. For a position $i$ of $w$, and a letter $a\in \Sigma$, let $d_{i}[a]=\Delta(\last_i[a],i)$.
%\begin{example}
%\end{example}

\begin{lemma}\label{init}
Let $w$ be a word, with $|w|=n$, $\letters(w)=\Sigma$, and $\Sigma = \{1,2,\ldots,\sigma\}$. We can compute in $O(n)$ the values $\Delta(1,\ell)$, for all $\ell \in [1:n]$. 
\end{lemma}
\begin{proof}
The pseudocode for this algorithm is given in \autoref{alg:delta}.

We define an array $C[1:\sigma]$, whose elements are initialised with $0$ and $f=0$. Now, we will traverse the positions of the word left to right.
When we reach position $\ell$, we do the following. 
If $C[w[\ell]]=0$, then we set $C[w[\ell]]=1$ and we increment $f$ by $1$. 
We set $\Delta(1,\ell)=f$.

	\begin{algorithm}[h]
		\SetKwInOut{Input}{Input}
		\SetKwInOut{Output}{Output}
		\Input{word $w$, alphabet $\Sigma$}
		\Output{$\Delta(i - \sigma + 1, i)$ for $i \in [\sigma:n]$}
		\BlankLine
		
		\tcp{initialization}
		int $f\leftarrow 0$; int $n \leftarrow \lvert w \rvert$; int $\sigma \leftarrow \lvert \Sigma \rvert$\;
		int $C[1:\sigma] = 0$\;
		\BlankLine
		
		\For{int $i \leftarrow 1$ \KwTo $n$}{
			\If{$C[w[i]] = 0$}{
				$f \leftarrow f + 1$\;
			}
			$C[w[i]] \leftarrow C[w[i]] + 1$\;
			\BlankLine
			$\Delta(1,i) \leftarrow f$\;
		}
		\caption{Calculation of $\Delta(1, i)$ for all $i \in [1:n]$}
		\label{alg:delta}
	\end{algorithm}
\end{proof}

\begin{lemma}\label{delta-sigma}
	Let $w$ be a word, with $|w|=n$, $\letters(w)=\Sigma$, and $\Sigma = \{1,2,\ldots,\sigma\}$. 
	We can compute in $O(n)$ the values $\Delta(i - \sigma + 1,i)$, for all $i \in [\sigma:n]$. 
\end{lemma}
\begin{proof}
The pseudocode for this algorithm is given in \autoref{alg:delta-sigma}.

	We define an array $C[1:\sigma]$, 
	whose elements are initialised with $0$ and $f=0$. 
	Now, we will traverse the positions of the word left to right as shown in \autoref{alg:delta-sigma}.
	
	When we reach position $i$, we do the following. 
	If $C[w[i]]=0$, we increment $f$ by $1$ since we saw a new letter.
	Then, $C[w[i]]$ is incremented by $1$. 
	When $i = \sigma$, we set $\Delta(1,i)=f$.
	When $i > \sigma$, we decrement $C[w[i-\sigma]]$ by one,
	and if the value of $C[w[i-\sigma]]$ is now $0$,
	meaning that it does not occur in $w[i-\sigma+1:i]$,
	we decrement $f$ by $1$.
	Finally, $\Delta(i-\sigma+1,i)$ is set to $f$.
	
	\begin{algorithm}[h]
		\SetKwInOut{Input}{Input}
		\SetKwInOut{Output}{Output}
		\Input{word $w$, alphabet $\Sigma$}
		\Output{$\Delta(i - \sigma + 1, i)$ for $i \in [\sigma:n]$}
		\BlankLine
		
		\tcp{initialization}
		int $f\leftarrow 0$; int $n \leftarrow \lvert w \rvert$; int $\sigma \leftarrow \lvert \Sigma \rvert$\;
		int $C[1:\sigma] = 0$\;
		\BlankLine
		
		\For{int $i \leftarrow 1$ \KwTo $n$}{
			\If{$C[w[i]] = 0$}{
				$f \leftarrow f + 1$\;
			}
			$C[w[i]] \leftarrow C[w[i]] + 1$\;
			\BlankLine
			
			\uIf{$i = \sigma$}{
				$\Delta(1,i) \leftarrow f$\;
			\BlankLine
			
			}\ElseIf{$i > \sigma$}{
				$C[w[i-\sigma]] \leftarrow C[w[i - \sigma]] - 1$\;
				\If{$C[w[i-\sigma]] = 0$}{
					$f \leftarrow f -1$\;
				}
				$\Delta(i - \sigma + 1, i) \leftarrow f$\;
			}
		}
		\caption{Calculation of $\Delta(i - \sigma + 1, i)$ for all $i \in [\sigma:n]$}
		\label{alg:delta-sigma}
	\end{algorithm}
\end{proof}

\begin{lemma}\label{last_occs_ins}
Let $w$ be a word, with $|w|=n$, $\letters(w)=\Sigma$, and $\Sigma = \{1,2,\ldots,\sigma\}$. We can compute in $O(n)$ time $\last_{j\sigma+1}[a]$ and $d_{j\sigma+1}[a]$, for all $a\in \Sigma$ and all integers $j\leq (n-1)/\sigma$.
\end{lemma}
\begin{proof}
The pseudocode for this algorithm is given in \autoref{alg:last-d}. 

We define an array $C[1:\sigma]$, whose elements are initialised with $\infty$, a variable $f=0$, as well as a doubly linked list $R$, which will have at most $\sigma$ elements from $[1:\sigma]$, and is empty at the beginning, and an array $P[1:\sigma]$ of pointers to the elements of $R$ (initially they are all set to an undefined value $\infty$). 

Now, we will traverse the positions of the word $w$ left to right. So, we consider each position $i$ of $w$, for $i$ from $1$ to $\left \lfloor \frac{n-1}{\sigma} \right \rfloor\sigma +1$ (the largest number of the form $j\sigma+1$, smaller or equal to $n$).

For position $i$, we do the following three steps. 

\begin{enumerate}
\item If $C[w[i]]=\infty$ then we increment $f$ by $1$. Both when $C[w[i]]=\infty$ or $C[w[i]]\neq \infty$ we set $C[w[i]] =i$. Now, $C[a]$ is the last occurrence of the letter $a\in \Sigma$ in the word $w[1:i]$, for all $a\in \Sigma$ (or $\infty$ if $a$ does not occur in $w[1:i]$). Also, $f$ is the number of elements of $C$ which are not equal to $\infty$, so the number of distinct letters we have seen in $w[1:i]$. 

\item If $w[i]$ does not occur in $R$ (tested by checking if $P[w[i]]=\infty$ or not), insert $w[i]$ at the end of $R$ (and set $P[w[i]]$ to point at the node where $w[i]$ occurs, which we have just created). If $w[i]$ occurs in $R$, remove $w[i]$ from $R$ (the place where $w[i]$ is stored in $R$ is $P[w[i]]$), insert $w[i]$ at the end of $R$, and update $P[w[i]]$ to point to this node of $R$. Now $R$ contains in order (from the front to the end) the $f$ distinct letters occurring in $w[1:i]$ ordered increasingly by the position of their last occurrence, and for each letter $a$ occurring in $w[1:i]$, $P[a]$ is a pointer to the node containing $a$ of $R$.

\item If $i=j\sigma +1$ for some $j$, we need to run the following two special steps. Firstly, we define the array $\last_{j\sigma+1}[\cdot]$, by setting $\last_{j\sigma+1}[a]=C[a]$ for all $a\in \Sigma$. Secondly, we set $g=f$; then, for each element $e$ in $R$, in the order in which these elements occur when traversing $R$ left to right, we set $d_{j\sigma+1}[e]=g$ and decrement $g$ by $1$. 
\end{enumerate}

It is not hard to see that the arrays $\last_{j\sigma+1}[\cdot]$ and $d_{j\sigma+1}[\cdot]$ are correctly computed for all $j\leq (n-1)/\sigma$. The overall time needed to maintain the array $C$ is linear, while computing each of the arrays $\last_{j\sigma+1}[1:\sigma]$ and $d_{j\sigma+1}[1:\sigma]$, for $j\leq n/\sigma$, takes $O(\sigma)$ time. However, as we only need to compute such arrays $O(n/\sigma)$ times, the overall time needed to compute them is $O(n)$. So, the statement holds.

\begin{algorithm}[h]
	\SetKwInOut{Input}{Input}
	\SetKwInOut{Output}{Output}
	\Input{word $w$, alphabet $\Sigma$}
	\Output{$\last_{j\sigma+1}$, $d_{j\sigma+1}$ for $j \leq (n-1)/\sigma$}
	\BlankLine
	
	\tcp{initialization}
	int $f = 0$; int $n \leftarrow \lvert w \rvert$; int $\sigma \leftarrow \lvert \Sigma \rvert$\;
	int $C[1:\sigma] = \infty$; doubly linked list $R$; pointer array $P[1:n]$\;
	\BlankLine
	
	\For{$i = 1$ \KwTo $\left \lfloor \frac{n-1}{\sigma} \right \rfloor\sigma +1$}{
		\tcp{step one}
		\If{$C[w[i]] == \infty$}{
			$f \leftarrow f + 1$\;
		}
		$C[w[i]] \leftarrow i$\;
		\BlankLine
		
		\tcp{step two}
		\If{$P[w[i]] \neq$ null}{
			remove the element in R at pointer $P[w[i]]$\;
		}
		add $w[i]$ to the right end of $R$\;
		$P[w[i]]$ $\leftarrow$ adress of $R[w[i]]$\;
        \BlankLine

        \tcp{step three}
	    \If{$i \mod \sigma = 1$}{
		$j \leftarrow (i-1)/\sigma$\;
		\ForEach{$a \in \Sigma$}{
		    $\last_{j\sigma+1}[a] \leftarrow C[a]$\;
	    }
        int $g \leftarrow f$\;
	    \ForEach{$e \in R$}{
	        $d_{j\sigma+1}[e] = g$\;
	        $g \leftarrow g - 1$\;
        }
      }
}
	\caption{Calculation of $\last_{j\sigma+1}$, $d_{j\sigma+1}$}
	\label{alg:last-d}
\end{algorithm}
\end{proof}

For $w=\tb \ta \tn \ta \tn \ta \tb \ta \tn $, we have $|w|=9$ and $\sigma=3$. In \autoref{init} we compute $\Delta(1,1)=1$, $\Delta(1,2)=2$, and $\Delta(1,\ell)=3$ for $\ell\in [3:9]$. In \autoref{delta-sigma} we compute $\Delta(1,3)=3$, $\Delta(2,4)=\Delta(3,5)=\Delta(4,6)=2$, $\Delta (5,7)=3$, $\Delta(6,8)=2$, and $\Delta(7,9)=3$. In \autoref{last_occs_ins} we compute the arrays $\last_{4}[\cdot]$ and $\last_{7}[\cdot]$. We get: $\last_{4}[\ta]=4$, $\last_{4}[\tb]=1$, $\last_{4}[\tn]=3$, and $\last_{7}[\ta]=6$, $\last_{7}[\tb]=7$, $\last_{7}[\tn]=5$. Therefore, $S_4=\{1,3,4\}$, $S_7=\{5,6,7\}$, and $d_4[\ta]=1$, $d_4[\tb]=3$, $d_4[\tn]=2$, $d_7[\ta]=2$, $d_7[\tb]=1$, $d_7[\tn]=3$.

For a word $w$ and a position $i$ of $w$, let $\univ[i]=\max\{j\mid w[j:i]$ is universal$\}$. That is, for the position $i$ we compute the shortest universal word ending on that position. If there is no universal word ending on position $i$ we set $\univ[i]=0$. 

Further, if $n=|w|$, let $V_w = \{\univ[i]\mid 1\leq i\leq n\}$. In $V_w$ we collect the starting positions of the shortest universal words ending at each position of the word $w$.
Now, for $j\in V_w$, let $L_j=\{i\mid \univ[i]=j\}$; in other words, we group together the positions $i$ of $w$ for which the shortest universal word ending on $i$ starts on some position $j$.
Note that $L_0=\{i\mid w[1:i]$ is not universal$\}$, i.e., the positions of $w$ where no universal word ends (see \autoref{fig_Ljs}).

Several observations are immediate: for $i\in L_j,~i'\in L_{j'}$, we have $i\leq i'$ if and only if~$j\leq j'$.  As each position $i$ of $w$ belongs to a set $L_j$, for some $j\in V_w$, we get that $\{L_j\mid j\in V_w\}$ is a partition of $[1:n]$ into intervals. 
Furthermore, $w[i]\neq w[j]$ for all $i\in L_j$ and $j\neq 0$: if $w[i]$ would be the same as $w[j]$ then $w[j+1:i]$ would also be a universal word, so $i$ would not be in $L_j$. 
Also, if $i=\max(L_j)$ for some $j>0$  then $w[i+1] = w[j]$. Indeed, there exists $j'\in[j+1:i]$ such that $w[j':i+1]$ is universal. But $w[j]$ does not occur in $w[j':i]$, so $w[j]=w[i+1]$ must hold.

Further, we define for all positions $i$ of $w$ the value $\freq[i]=|w[1:i]|_{w[i]}$, the number of occurrences of $w[i]$ in $w[1:i]$. Also, let $T[i]=\min \{|w[i+1:n]|_{a} \mid a \in \Sigma \}$, for $i\in [0,n-1]$,  
be the least number of occurrences of a letter in $w[i+1:n]$; set $T[n]=0$. 
%%% $\prev_{w[i]}[j] = \last_{j-1}[w[i]], for all $i \in L_j$; 

%\begin{example}
%\end{example}

\begin{lemma}\label{help1}
Let $w$ be a word, with $|w|=n$, $\letters(w)=\Sigma$, and $\Sigma = \{1,2,\ldots,\sigma\}$. 
We can compute in $O(n)$ time the following data structures:
%\begin{enumerate}
%\item 
$1.$ the array $\univ[\cdot]$; 
%\item 
$2.$ the set $V_w$ and the lists $L_j$, for all $j\in V_w\setminus\{0\}$;
%\item 
$3.$ the array $\freq[\cdot]$;
%\item 
$4.$ the array $T[\cdot]$;
%\item 
$5.$ the values $\last_{j-1}[w[i]]$, for all $j\in V_w$ and all $i \in L_j$; 
%\item 
$6.$ the values $\last_{i-1}[w[i]]$, for all $i \in [2:n]$.
%\end{enumerate}
\end{lemma}
\begin{proof}
We present first an algorithm for items $1$ and $2$, then an algorithm for items $3$ and $4$, and finally an algorithm for items $5$ and $6$. 
\begin{figure}	
	\tikzset{every picture/.style={line width=0.75pt}} %set default line width to 0.75pt     
	
	\begin{tikzpicture}[x=0.75pt,y=0.75pt,yscale=-1,xscale=1]
	
	%Straight Lines [id:da32061185263821934] 
	\draw    (117.5,120) -- (540,120) ;
	%Curve Lines [id:da4119091411923843] 
	\draw    (200,100) .. controls (252,-4.67) and (438.5,19.5) .. (480,100) ;
	%Curve Lines [id:da497749575774145] 
	\draw    (201.95,96.69) .. controls (245.37,26.08) and (333.43,72.42) .. (340,100) ;
	\draw [shift={(200,100)}, rotate = 299.53] [fill={rgb, 255:red, 0; green, 0; blue, 0 }  ][line width=0.08]  [draw opacity=0] (10.72,-5.15) -- (0,0) -- (10.72,5.15) -- (7.12,0) -- cycle    ;
	%Straight Lines [id:da2633429172574029] 
	\draw    (200,110) -- (200,130) ;
	%Straight Lines [id:da26907568938345117] 
	\draw [color={rgb, 255:red, 83; green, 83; blue, 83 }  ,draw opacity=1 ][line width=3]    (340,120) .. controls (341.67,118.33) and (343.33,118.33) .. (345,120) .. controls (346.67,121.67) and (348.33,121.67) .. (350,120) .. controls (351.67,118.33) and (353.33,118.33) .. (355,120) .. controls (356.67,121.67) and (358.33,121.67) .. (360,120) .. controls (361.67,118.33) and (363.33,118.33) .. (365,120) .. controls (366.67,121.67) and (368.33,121.67) .. (370,120) .. controls (371.67,118.33) and (373.33,118.33) .. (375,120) .. controls (376.67,121.67) and (378.33,121.67) .. (380,120) .. controls (381.67,118.33) and (383.33,118.33) .. (385,120) .. controls (386.67,121.67) and (388.33,121.67) .. (390,120) .. controls (391.67,118.33) and (393.33,118.33) .. (395,120) .. controls (396.67,121.67) and (398.33,121.67) .. (400,120) .. controls (401.67,118.33) and (403.33,118.33) .. (405,120) .. controls (406.67,121.67) and (408.33,121.67) .. (410,120) .. controls (411.67,118.33) and (413.33,118.33) .. (415,120) .. controls (416.67,121.67) and (418.33,121.67) .. (420,120) .. controls (421.67,118.33) and (423.33,118.33) .. (425,120) .. controls (426.67,121.67) and (428.33,121.67) .. (430,120) .. controls (431.67,118.33) and (433.33,118.33) .. (435,120) .. controls (436.67,121.67) and (438.33,121.67) .. (440,120) .. controls (441.67,118.33) and (443.33,118.33) .. (445,120) .. controls (446.67,121.67) and (448.33,121.67) .. (450,120) .. controls (451.67,118.33) and (453.33,118.33) .. (455,120) .. controls (456.67,121.67) and (458.33,121.67) .. (460,120) .. controls (461.67,118.33) and (463.33,118.33) .. (465,120) .. controls (466.67,121.67) and (468.33,121.67) .. (470,120) .. controls (471.67,118.33) and (473.33,118.33) .. (475,120) .. controls (476.67,121.67) and (478.33,121.67) .. (480,120) -- (480,120) ;
	%Straight Lines [id:da6536987298156314] 
	\draw    (340,110) -- (340,130) ;
	%Straight Lines [id:da757591923394991] 
	\draw    (480,110) -- (480,130) ;
	
	% Text Node
	\draw (91,117) node [anchor=north west][inner sep=0.75pt]   [align=left] {$\displaystyle w$};
	% Text Node
	\draw (475,140) node [anchor=north west][inner sep=0.75pt]   [align=left] {$\displaystyle i_{1}$};
	% Text Node
	\draw (152,142) node [anchor=north west][inner sep=0.75pt]   [align=left] {$\displaystyle univ[ i] =j$};
	% Text Node
	\draw (335,140) node [anchor=north west][inner sep=0.75pt]   [align=left] {$\displaystyle i_{0}$};

	\end{tikzpicture}
\caption{{ The set of elements with $\univ[i]=j$ forms an interval: $ L_{j} \ =[i_{0}:\ i_{1}]$.}} \label{fig_Ljs}
\end{figure}
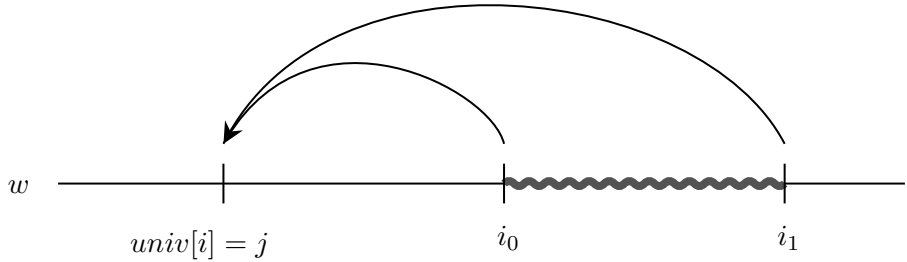

{\em Algorithm for 1,2.} (pseudocode: \autoref{alg:univ-v-l}) To compute $\univ$, $V_w$, and the lists $L_j$ we will use a {\em two pointer-strategy}. That is, we go through the positions of $w$ from $n$ to $1$ with two pointers $a$ and $b$. Initially, $a=n$ and $b=n+1$. We also define an array $C$ with $\sigma$ elements, all initially set to $0$, and a variable $f$, initialised with $0$. The elements of the array $\univ$ are initialised with $0$.

Now we repeat the following three-step procedure until $b$ is $0$. 

In the first step, we execute the following loop. While $b>1$ and $f<\sigma$, we do the following:  decrement $b$ by $1$, increment $C[w[b]]$ by $1$ and if $C[w[b]]=1$ then increment $f$ by $1$, too. 

In the second step, when the loop finished, if $f<\sigma$ and $b=1$, then set $b=0$. 
If $f=\sigma$, then we add $b$ to $V_w$.

In the third step, while $f=\sigma$, we do the following steps. First, set $\univ[a]=b$ and store $a$ in $L_b$.  Decrement $C[w[a]]$ by $1$. Now, if $C[w[a]]=0$, then decrement $f$ by $1$, too. Finally, decrement $a$ by $1$. 

The idea is relatively simple: we start with the factor $w[a]$ and then try to extend it to the left (i.e., produce the factors $w[b:a]$ with $b\leq a$) while keeping track in $f$ how many different letters of $\Sigma$ we met (i.e., $f=\Delta[b:a]$), and their respective counts in the array $C[\cdot]$. As soon as we have seen all letters (i.e., $f=\sigma$), we know that $w[b:a]$ is the shortest universal word ending on $a$. Thus, we can store $b$ in $V_w$, and set $\univ[a]=b$. Now, we try to identify all the universal words starting on $b$, ending to the left of $a$ (their respective ending positions and $a$ are exactly the elements of $L_b$). This is done similarly: we move the pointer $a$ now to the left one position at a time, and see which letters of $\Sigma$ are removed from the word $w[b:a]$, using the array $C$ where we counted the letter-occurrences. As soon as we have a letter that occurs $0$ times in $w[b:a]$ we stop, as $w[b:a]$ is no longer universal. We then repeat the procedure: move $b$ to the left first till we again have that $w[b:a]$ is universal, then move $a$ to the left till $w[b:a]$ is no longer universal, and so on. By the observations we made on the structure of the lists $L_j$, this approach is clearly correct. 

The complexity is linear, as each pointer visits each position of $w$ once. 

\SetKw{KwAnd}{and}
\SetKwFunction{add}{add}
\begin{algorithm}[h!] 
	\SetKwInOut{Input}{Input}
	\SetKwInOut{Output}{Output}
	\Input{word $w$, alphabet $\Sigma$}
	\Output{$univ$, $V_w$, $L_j$ for $j \leq \sigma$}
	\BlankLine
	
	\tcp{initialization}
	int $f\leftarrow 0$; int $a \leftarrow n$; int $b\leftarrow n+1$; int $n \leftarrow \lvert w \rvert$; int $\sigma \leftarrow \lvert \Sigma \rvert$\;
	int $C[1:\sigma] = 0$, $\univ[1:n] = 0$\;
	\BlankLine
	
	\While{$b > 0$}{
	\tcp{step one}
	\While{$f < \sigma$ \KwAnd $b > 1$}{
	    $b \leftarrow b -1$\;
	    \If{$C[w[b]] = 0$}{
	        $f \leftarrow f + 1$\;
        }
	    $C[w[b]] \leftarrow C[w[b]] + 1$\;
    }
    \BlankLine
    
    \tcp{step two}
    \uIf{$f = \sigma$}{
    	$V_w.\add(b)$\;
    }\ElseIf{$f < \sigma$ \KwAnd $b = 1$}{
        $b \leftarrow 0$\;
    }
    \BlankLine
    
    \tcp{step three}
    \While{$f = \sigma$}{
    	$univ[a] = b$\;
        $L_b.\add(a)$\;
        $C[w[a]] \leftarrow C[w[a]] - 1$\;
        \If{$C[w[a]] = 0$}{
        	$f \leftarrow f - 1$\;
        }
        $a \leftarrow a - 1$\;
    }}
	\caption{Calculation of $univ$, $V_w$, $L_j$}
	\label{alg:univ-v-l}
\end{algorithm}

{\em Algorithm for 3,4.}  (pseudocode: \autoref{alg:freq-t}) To compute $\freq[\cdot]$ we use a straightforward strategy. We use an array $C[\cdot]$ with $\sigma$ elements, initially set to $0$. Then, for $i$ from $1$ to $n$ we increment $C[w[i]]$ by $1$ and set $\freq[i]=C[w[i]]$. This clearly takes linear time. At the end of this traversal of the word, $C[a]$ is the number of occurrences of $a$ in $w$. We will show now how the array $T$ is computed. Let $x$ be the letter of $\Sigma$ such that $C[x]$ is the smallest value of $C$ and $m=C[x]$. We set $T[0]=m$. Now, for $i$ from $1$ to $n-1$ we do the following three steps. Firstly, we decrement $C[w[i]]$ by $1$. Secondly, if $C[w[i]]<m$ then we set $m=C[w[i]]$ and $x=w[i]$. Thirdly, we set $T[i]=m$. It is immediate that $T$ is correctly computed and that the computation takes linear time.

\SetKw{KwAnd}{and}
\SetKwFunction{add}{add}
\begin{algorithm}[h!]
	\SetKwInOut{Input}{Input}
	\SetKwInOut{Output}{Output}
	\Input{word $w$, alphabet $\Sigma$}
	\Output{$freq$, $T$}
	\BlankLine
	
	\tcp{initialization}
	int $n \leftarrow \lvert w \rvert$; int $\sigma \leftarrow \lvert \Sigma \rvert$\;
	int $C[1:\sigma] = 0$; int $freq[n]$; int $T[n]$\;
	\BlankLine
	
	\For{$i \leftarrow 1$ \KwTo $n$}{
	    $C[w[i]] \leftarrow C[w[i]] + 1$\;
        $freq[i] = C[w[i]]$\;
    }
	\BlankLine
	
	int $x$\;
	int $m \leftarrow n+1$\;
	\For{$i \leftarrow 1$ \KwTo $\sigma$}{
	    \If{$C[w[i]] < m$}{
            $m \leftarrow C[w[i]]$\;
            $x \leftarrow w[i]$\;
        }
    }
	\BlankLine
	
	$T[0] \leftarrow m$\;
	\For{$i \leftarrow 1$ \KwTo $n-1$}{
		\tcp{step one}
		$C[w[i]] \leftarrow C[w[i]] -1$\;
		\BlankLine
		
		\tcp{step two}
		\If{$C[w[i]] < m$}{
			$m \leftarrow C[w[i]]$\;
			$x \leftarrow w[i]$\;
		}
		\BlankLine
		
		\tcp{step three}
		$T[i] \leftarrow m$\;
	}
	\caption{Calculation of $freq$ and $T$}
	\label{alg:freq-t}
\end{algorithm}

{\em Algorithm for 5,6.} (pseudocode: \autoref{alg:last}) For the computation of all the values $\last_{j-1}[w[i]]$, for all $j\in V_w$ and all $i \in L_j$, and the values $\last_{i-1}[w[i]]$, for all $i \in [2:n]$, we do the following. We use an array $L[\cdot]$ with $\sigma$ elements, initially set to $n+1$. Then, for $i$ from $0$ to $n$ we do the following two steps. If $i>0$, set $\last_{i-1}[w[i]]=L[w[i]]$ and $L[w[i]]=i$. If $i+1\in V_w$, then we go through the elements $e$ of list $L_{i+1}$ and set $\last_{i}[w[e]]=L[w[e]]$. This takes linear time, as the time needed to execute the iteration of the loop for some $i$ is either $O(1)$ if $i+1\notin V_w$ or $1+O(|L_{i+1}|)$ if $i+1\in V_w$. This adds up to $O(\sum_{j\in V_w}|L_j|)=O(n)$.

It is important to note here that each $\last_i[\cdot]$ is implemented as a list (implemented statically) with exactly one element if $i+1\notin V_w$ (i.e., $\last_i[w[i+1]]$) and exactly $|L_{i+1}\cup\{w[i+1]\}|$ elements (i.e., $\last_i[w[i+1]]$ and $\last_i[w[j]]$ with $j\in L_{i+1}$). In the second case, the list is implemented as an array, indexed by the the letters in $L_{i+1}\cup \{w[i+1]\}$.

\SetKw{KwAnd}{and}
\SetKwFunction{add}{add}
\begin{algorithm}[h!]
	\SetKwInOut{Input}{Input}
	\SetKwInOut{Output}{Output}
	\Input{word $w$, alphabet $\Sigma$, $V_w$, lists $L_j$ for $j \in V_w$}
	\Output{values $\last_{j-1}[w[i]]$ for all $j \in V_w$ and $i \in L_j$, values $\last_{i-1}[w[i]]$ for all $i \leq n$}
	\BlankLine
	
	\tcp{initialization}
	int $n \leftarrow \lvert w \rvert$; int $\sigma \leftarrow \lvert \Sigma \rvert$\;
	int $L[1:\sigma] = n+1$\;
	\BlankLine
	
	\For{$i \leftarrow 0$ \KwTo $n-1$}{
		\If{$i + 1 \in V_w$}{
			\ForEach{$e \in L_{i+1}$}{
				$\last_i[w[e]] \leftarrow L[w[e]]$\;
			}
		}
		\BlankLine
		
		\If{$i < n$}{
			$\last_{i}[w[i+1]] \leftarrow L[w[i+1]]$\;
			$L[w[i+1]] \leftarrow i+1$\;
		}
	}
	\caption{Calculation of $\last_{i}[w[j]]$ for $i+1 \in V_w$ and $j \in L_{i+1}$, and $\last_{i}[w[i+1]]$ for $i \leq n$.}
	\label{alg:last}
\end{algorithm}

This concludes our proof.
 \end{proof}
 Consider again $w=\tb \ta \tn \ta \tn \ta \tb \ta \tn$. In \autoref{help1} we compute the following values. Firstly, $\univ[1]=\univ[2]=0$, $\univ[\ell]=1$ for $\ell\in [3:6]$, $\univ[7]=\univ[8]=5$, $\univ[9]=7$. Thus, $V_w=\{0,1,5,7\}$ and $L_0=[1:2]$, $L_1=[3:6]$, $L_5=[7:8]$, $L_7=[9:9]$. Secondly, $\freq[1]=1$, $\freq[2]=\freq[3]=1$, $\freq[4]=\freq[5]=2$, $\freq[6]=3$, $\freq[7]=2$, $\freq[8]=4$, $\freq[9]=3$. Moreover, $T[0]=2$, $T[\ell]=1$ for $\ell \in [1:6]$, and $T[\ell]=0$ for $\ell \in [7:9]$. Then, for $j=1$, we have $\last_0[a]=0$, for $a\in \{\ta,\tb,\tn\}$; for $j=5$, we have $\last_{4}[\tb]=1$ and  $\last_{4}[\ta]=4$; for $j=7$, we have $\last_{6}[\tn]=5$. Finally, $\last_0[\tb]=10$, $\last_1[\ta]=10$, $\last_2[\tn]=10$, $\last_{3}[\ta]=2$, $\last_4[\tn]=3$, $\last_5[\ta]=4$, $\last_6[\tb]=1$, $\last_7[\ta]=6$, $\last_8[\tn]=5$.
 
%The main idea behind proving Lemmas \ref{init}, \ref{last_occs_ins}, and \ref{help1} is to traverse the word $w$ left to right (or, respectively, right to left) and maintain the number of occurrences, as well as the last occurrence, of each letter in the prefix (respectively, suffix) of $w$ that we have visited so far. For Lemma \ref{delta-sigma}, we only consider a sliding window of size $\sigma$ which traverses the word left-to-right, while maintaining similar data as before, but only for the content of the window. In all cases, this requires linear time and enables us to construct the desired data structures. 

Together with the string-processing data structures we defined above, we need the following general technical data structures lemma. This lemma (combined with some combinatorial observations) will be used to speed up some of our dynamic programming algorithms.

In this lemma we process a list $A$ which initially has $\sigma$ elements, and in which we insert, in successive steps, $\sigma$ new elements, by appending them always at the same end. 
For simplicity, we can assume that the list $A$ is a sequence with $2\sigma$ elements (denoted $A[i]$, with $i\in [1:2\sigma]$), out of which the last $\sigma$ are initially undefined. The \nth{$i$} insertion would, consequently, mean setting $A[\sigma+i]$ to the actual value that we want to insert in the list $A$. In our lemma we will also repeatedly perform an operation which decrements the values of some elements of the list $A$. However, we will not require to be able to explicitly access, after every operation, all the elements of the list (so we will not need to retrieve the values $A[i]$). Consequently, we will not maintain explicitly the value of all the elements of $A$ (that is, we will not update the elements affected by decrements). We are only interested in being able to retrieve  (by value and position), at each moment, the smallest element and the last element of $A$. Thus, throughout the computation, we only maintain a subset of important elements of $A$, including the aforementioned two. We can now state our result, whose proof is based on \autoref{LemUnionFind}.

\begin{lemma}\label{speedUp}
Let $A$ be a list with $\sigma$ elements (natural numbers) and let $m=\sigma$. We can execute (in order) the sequence of $\sigma$ operations $o_1,\ldots, o_\sigma$ on $A$ in overall $O(\sigma)$ time, where $o_i$ consists of the following three steps, for $i\in [1:m]$:
\begin{enumerate}
\item Return $e=\argmin\{A[i]\mid i\in [1:m]\}$ and $A[e]$.
\item For some $j_i\in [1:m]$, decrement all elements $A[j_i], A[j_i+1], \ldots, A[m]$ by $1$.
\item For some natural number $x_i$, append the element $x_i$ to $A$ (i.e., set $A[m+1]$ to $x_i$), and increment $m$ by $1$ (i.e., set $m$ to $m+1$). 
\end{enumerate}
\end{lemma}
\begin{proof}
Firstly, we will {\em run a preprocessing} of $A$.

We begin by defining recursively a finite sequence of positions as follows:
\begin{itemize}
\item $a_1$ is the rightmost position of $A$ on which $\min\{A[i]\mid i\in [1:\sigma]\}$ occurs;
\item for $i\geq 2$, if $a_{i-1}<\sigma$, then $a_i$ is the rightmost position on which $\min\{A[i]\mid i\in [a_{i-1}+1:\sigma]\}$ occurs; 
\item for $i\geq 2$, if $a_{i-1}=\sigma$, then we can stop, our sequence will have $i-1$ elements.
\end{itemize}

Let $p$ be the number of elements in the sequence defined above, i.e., our sequence is $a_1,\ldots,a_p$. For convenience, let $a_0=0$. Then the sequence $a_1,\ldots,a_p$ fulfils the following properties:
\begin{itemize}
\item $a_p=\sigma$ and $a_i>a_{i-1}$, for all $i\in [1:p]$;
\item $A[a_i] > A[a_{i-1}]$ for all $i\in [2:p]$;
\item for all $i\in [1:p]$, we have $A[a_i]<A[t]$, for all $t\in [a_i+1:\sigma]$;
\item for all $i\in [1:p]$, we have $A[a_i]\leq A[t]$, for all $t\in [a_{i-1}+1:a_i]$.
\end{itemize}
By definition, for $i\in [1:p]$ we have $A[a_i]=\min\{A[t]\mid t\in [a_{i-1}+1:\sigma]\}$, $A[a_i]<\min\{A[t]\mid t\in [a_{i}+1:\sigma]\} $, and $a_1=\min\{A[i]\mid i\in[1:\sigma]\}$. Clearly, we have $a_p=\sigma$.

The positions $a_1,\ldots,a_p$ can be computed in linear time $O(\sigma)$, in reversed order. As we do not know from the beginning the value of $p$, we will compute a sequence $b_1, b_2, \ldots$ of positions as follows. We start with $b_1=\sigma$, $t=\sigma-1$, and $i=2$. Then, while $t\geq 1$ we do the following case analysis. If $A[t]<b_{i-1}$, then set $b_i=t$, increment $i$ by $1$, and decrement $t$ by $1$.  Otherwise, if $A[t]\geq b_{i-1}$, just decrement $t$ by $1$. It is straightforward that this process takes $O(\sigma)$ time, and, when we have finished it, the number $i$ is exactly the number $p$, and $a_i=b_{p-i+1}$.

Another observation is that, for $a_0=0$, the intervals $[a_{i-1}+1,a_{i}]$, for $i\in[1,p]$, define a partition of the interval $[1:\sigma]$ into $p$ intervals. Therefore, we can define a partition of the interval $[1:2\sigma]$ into the intervals $[a_{i-1}+1:a_{i}]$, for $i\in[1,p]$, and $[t:t]$, for $t\in [\sigma+1:2\sigma]$. Thus, we construct in linear time, according to \autoref{UnionFind}, an interval union-find data structure for the interval $[1:2\sigma]$, as induced by the intervals $[1:a_1]$, $[a_{1}+1:a_{2}], \ldots$ $[a_{p-1},a_p]$, $[\sigma+1:\sigma+1]$, $[\sigma+2:\sigma+2], \ldots$ $[2\sigma:2\sigma]$. 

Let us now take $m=\sigma$ (and assume the convention $A[0]=0$). We associate as satellite data to each interval $[x:y]$ with $y\leq m$ from our interval union-find data structure the value $A[y]-A[x-1]$. 

This entire preprocessing takes clearly $O(\sigma)$ time.

In order to explain how the operations are implemented, we assume as invariant that the following properties are fulfilled before $o_i$ is executed, for $i\in[1:\sigma]$:
\begin{itemize}
\item $A$ contains $m$ elements;
\item all intervals $[x:y]$ with $y>m$ from our interval union-find data structure are singletons (i.e., $x=y$);
\item for each interval $[x:y]$ with $y\leq m$, we have the associated satellite data $A[y]-A[x-1]$;
\item for each interval $[x:y]$ with $y\leq m$, we have that $A[y]\leq A[t]$ for $t\in [x:m]$ and $A[y]< A[t]$ for $t\in [y+1:m]$;
\item we have stored in a variable $\ell$ the value $A[m]$.
\end{itemize}

This clearly holds after the preprocessing step, so before executing $o_1$. 

Let us now explain {\em how the operation $o_i$ is executed}. 

{\em The first step} of $o_i$ is to return $e=\min\{A[i]\mid i\in [1:m]\}$ and $i_e$ the rightmost position of the list $A$ such that $A[i_e]=e$. We execute $\find(1)$ to return the first interval $[1:i_e]$ stored in our interval union-find data structure; $A[i_e]$ is the satellite data associated to this interval (by convention, $A[i_e]-A[1-1]=A[i_e]-A[0]=A[i_e]$). The fact that the invariant property holds shows that $i_e$ is correctly computed.

{\em The second step} of $o_i$ is to decrement all elements $A[j_i], A[j_i+1], \ldots, A[m]$ by $1$, for some $j_i \in [1:m]$. We will make no actual change to the elements of the list $A$, as this would be too inefficient, but we might have to change the state of the union-find data structure, as well as the satellite data associated to some intervals of this structure.

So, let $[x:y]$ be the interval containing $j_i$, returned by $\find(j_i)$, and assume first that $x\neq 1$. 

According to the invariant, $A[j_i]\geq A[y]$ and $A[y]>A[x-1]$. After decrementing the elements $A[j_i], A[j_i+1], \ldots, A[m]$ by $1$, the difference $A[t]-A[t']$ is exactly the same as before, for all $t,t'\in [j_i:m]$. In consequence, the relative order between the elements of the suffix $A[j_i:m]$ of the list $A$ is preserved. Also, for all $t\in [x:j_i-1]$, we have now $A[t]>A[y]$ (before decrementing $A[y]$ we had only $A[t]\geq A[y]$). However, the difference $A[y]-A[x-1]$ is now decreased by $1$. If it stays strictly positive, we just update the satellite data of the respective interval (by decrementing it accordingly by $1$). If $A[y]-A[x-1]=0$, then we make the $\union$ of the interval $[z:x-1]$ (returned by $\find(x-1)$) and $[x:y]$ to obtain the new interval $[z:y]$. Its satellite data is $A[y]-A[z-1]=A[x-1]-A[z-1]$, so the same as the satellite data that was before associated to $[z:x-1]$. The invariant is clearly preserved, as, even after decrementing it, $A[y]$ (which is now equal to $A[x-1]$) is strictly greater than $A[z-1]$, strictly smaller than $A[t]$, for $t\in [y+1:m]$, and smaller than or equal to $A[t]$, for $t\in [z:y]$. 

If the interval containing $j_i$ is $[1:y]$, then we just update the satellite data of the respective interval by decrementing it by $1$. 

{\em The third step} of $o_i$ is to append the element $x_i$ to $A$ (i.e., set $A[m+1]=x_i$), for some natural number $x_i$, and increment $m$ by $1$. 

We implement this as follows. Let $t=m$ and $q=A[m]$ (this value is stored and maintained using the variable $\ell$). While $t\geq 1$ do the following. Let $[z,t]$ be the interval returned by $\find (t)$; we have $q=A[t]$. If $q\geq x_i$, make the union of $[z:t]$ and $[t+1:m+1]$; update $q=q-(A[t]-A[z-1])=A[z-1]$ (using the satellite data $A[t]-A[z-1]$ associated to $[z,t]$), update $t=z-1$, and reiterate the loop.  If $q< x_i$, exit the loop. After this, we set $m$ to $m+1$ and $\ell=x_i$. 

It is not hard to see that after running this third step, so before executing operation $o_{i+1}$, the invariant is preserved. 

Performing operation $o_i$ takes an amount of time proportional to the sum of the number of $\union$ and the number of $\find $ operations executed during its three steps. By \autoref{LemUnionFind}, this means that executing all operations $o_1,\ldots, o_\sigma$ takes in total at most $O(\sigma)$ time. 
\end{proof} 
%The idea of the proof is the following. We maintain (update after each operation), using Lemma \ref{LemUnionFind}, a partition of $[1:2\sigma]$ into intervals $[1:a_1], [a_1+1:a_2],$ $\ldots, [a_p:m], [m+1:m+1], \ldots, [2\sigma:2\sigma]$, such that, for all $i\in [1:p]$, $A[a_i]$ is strictly smaller than all elements of $A[a_i+1:m]$ and smaller or equal to the elements of $A[a_{i-1}+1:a_i]$ (with $a_0=0$). The set of the elements $A[a_i]$, for $i\in [1:p]$, is also maintained. With this,~$A[a_1]$ is the minimum of $A$, and we can implement each operation in amortised $O(1)$ time. 

\input{exampleToolbox}

%%%%%%%%%%%%%%%%%%%%%%%%%%%%

\section{Edit Distance}\label{editDist}

We are interested in computing the minimal number of edit operations we need to apply to a word $w$, with $|w|=n$, $\letters(w)=\Sigma$, with universality index $\iota(w)$, so that it is transformed into a word with universality index $k$, w.r.t. the same alphabet $\Sigma$. The edit operations considered are the usual ones (insertion, deletion, substitution), and the number we want to compute can be seen as the {\em edit distance} between $w$ and the set of $k$-universal words over $\Sigma$. 

However, when we want to obtain a $k$-universal word with $k>\iota(w)$, then it is enough to consider only insertions. Indeed, deleting a letter of a word can only restrict the set of subsequences of the respective word, while in this case we are interested in enriching it. Substituting a letter might make sense, but it can be simulated by an insertion: assume one wants to substitute the letter $\ta$ on position $i$ of a word $w$ by a $\tb$. It is enough to insert a $\tb$ next to position $i$, and the set of subsequences of $w$ is enriched with all the words that could have appeared as subsequences of the word where $\ta$ was actually replaced by $\tb$. We might have some extra words in the set of subsequences, which would have been eliminated through the substitution, but it does not affect our goal of reaching $k$-universality. So, to increase the universality index of a word it is enough to use insertions.

When we want to obtain a word with universality index $k$, for $k<\iota(w)\leq n/\sigma$, then it is enough to consider only deletions. Assume that we have a sequence of edit operations that transforms the word $w$ into a word $w'$ with universality index $k$. Now, remove all the insertions of letters from that sequence. The word $w''$ we obtain by executing this new sequence of operations clearly fulfils $\iota(w'')\leq \iota(w')$. Further, in the new sequence, replace all substitutions with deletions. We obtain a word $w'''$ with a set of subsequences strictly included in the one of $w''$, so with $\iota(w''')\leq \iota(w'')$. As each deletion changes the universality index by at most $1$, it is clear that (a prefix of) this new sequence of deletion operations witnesses a shorter sequence of edit operations which transforms $w$ into a word of universality index $k$. Thus, to decrease the universality index of a word it is enough to use deletions.

Finally, in a third case, one might be interested in what happens if we only 
use substitutions. In this way, we can both decrease and increase the 
universality index of a word. Moreover, one can see the minimal number of 
substitutions needed to transform $w$ into a $k$-universal word as the {Hamming distance} between $w$ and the set of $k$-universal words.

We will discuss each of these cases separately. 

\subsection{Insertions}\label{inst}
%%%%%%%%%%%%%%%%%%%%%%%%%%%%%%%%%%%%%%%%%%%%%%%%%%%

\begin{theorem}\label{ins-distance}
Let $w$ be a word, with $|w|=n$, $\letters(w)=\Sigma$, and $\Sigma = \{1,2,\ldots,\sigma\}$. Let $k\geq \iota(w)$ be an integer. We can compute the minimal number of insertions needed to apply to $w$ in order to obtain a $k$-universal word (w.r.t. $\Sigma$) in $O(nk)$~time if $k\leq n$ and $O(n^2 + T(n,\sigma,k))$ time otherwise, where $T(n,\sigma,k)$ is the time needed to compute $(k-n)\sigma $. 
\end{theorem}
\begin{proof}
{\bf Case 1.} Let us assume first that $k\leq n$. We structured our proof in such a way that the idea of the solution, as well as the actual computation steps, and the arguments supporting their correctness are clearly marked. See also \autoref{alg:insertions}.

\smallskip

{\emph{$\S$ General approach.}} We want to transform the word $w$ into a $k$-universal word with a minimal number of insertions. Assume that the word we obtain this way is $w'$, and $|w'|=m$. Thus, $w'$ has a prefix $w'[1:m']$ which is $k$-universal, but $w'[1:m'-1]$ is not $k$-universal. Moreover, $w'[1:m']$ is obtained from a prefix $w[1:\ell]$ of $w$, and $w'[m'+1:m]=w[\ell+1:n]$. Indeed, any insertion done to obtain $w'[m'+1:m]$ can be simply omitted and still obtain a $k$-universal word from $w$, with a lower number of insertions. 

Consequently, it is natural to compute the minimal number of insertions needed to transform $w[1:\ell]$ into a $t$-universal word, for all $\ell \leq n$ and $t\leq k$. Let $M[\ell][t]$ denote this number. By the same reasoning as above, transforming (with insertions) $w[1:\ell]$ into a $t$-universal word means that there exists a prefix $w[1:\ell']$ of $w[1:\ell]$ which is transformed into a $(t-1)$-universal word and $w[\ell'+1:\ell]$ is transformed into a $1$-universal word. Clearly, the number of insertions needed to transform $w[\ell'+1:\ell]$ into a $1$-universal word is $\sigma-\Delta(\ell'+1,\ell)$, i.e., the number of distinct letters not occurring in $w[\ell'+1:\ell]$. As we are interested in the minimal number of insertions needed to transform $w[1:\ell]$ into a $t$-universal word, we need to find a position $\ell'$ such that the total number of insertions needed to transform $w[1:\ell']$ into a $(t-1)$-universal word and $w[\ell'+1:\ell]$ into a $1$-universal word is minimal.

\smallskip

\SetKw{KwAnd}{and}
\begin{algorithm}[h!]
	\SetKwInOut{Input}{Input}
	\SetKwInOut{Output}{Output}
	\Input{word $w$, alphabet $\Sigma$, int $k$}
	\Output{minimal number of insertions}
	\BlankLine
	
	\tcp{initialization}
	int $n \leftarrow \lvert w \rvert$; int $\sigma \leftarrow \lvert \Sigma \rvert$;  int $d\gets 0$\;
	int $M[n][k]$\;
	\BlankLine
	
	\tcp{initialise first column of $M$}
	\For{$l = 1$ \KwTo $n$}{
		$M[l][1] \leftarrow \sigma - \Delta(1,l)$\;
	}
	\BlankLine
	
	\tcp{efficient variant}
	\For{$t = 2$ \KwTo $k$}{
		\tcp{$\leq (n-1)/\sigma$ phases}
		\For{$j=0$ \KwTo $\lceil (n-1)/\sigma\rceil $}{
			int $A[\sigma][3]$ (list of triples including satellite data); int $\pos[\sigma]$\;
			\For{$a = 1$ \KwTo $\sigma$}{
				\If{$d_{j\sigma + 1}[a] > 0$}{
					int $i \leftarrow \sigma - d_{j\sigma + 1}[a]$; $d\gets d+1$\; 
					$A[i+1][1] \leftarrow M[\last_{j\sigma+1}[a]-1][t-1]+i$\;
					$\pos[a] \leftarrow i + 1$\;
					\tcp{satellite data for $A[i+1]$}
					$A[i+1][2] \leftarrow \last_{j\sigma+1}[a]$\;
					$A[i+1][3] \leftarrow a$\;
				}
				
				\If{$d_{j\sigma + 1}[a] = 0$ \KwAnd $\last_{j\sigma+1}[a] = n+1$}{
					$\pos[a] = 0$\;
				}
			}
			\BlankLine
			
			\tcp{exactly $d$ letters of $\Sigma$,  occur in $w[1:j\sigma+1]$}
			\tcp{(ordered increasingly by their last occurrence):}
			\tcp{ $A[\sigma-d+1][3], A[\sigma-d+2][3],\ldots, A[\sigma-1][3], A[\sigma-][3].$ }
			set all elements in $A[1:\sigma - d]$ to $\infty$ (they cannot be changed)\;
			$m = \sigma$\;
			\BlankLine
			
			\tcp{apply sequence of operations as in \autoref{speedUp}}
			\For{$i = 1$ \KwTo $\sigma$}{
				$q \leftarrow$ minimum of $A$\;
				$M[j\sigma+i][t] \leftarrow \min\{q, M[j\sigma +i][t-1]+\sigma\}$\;
				$a = w[j\sigma + i + 1]$\;
				decrement positions $\pos[a]+1, \pos[a]+2, \ldots, m$ by $1$\;
				append $M[j\sigma + i][t-1]+(\sigma -1)$ to $A$ (i.e., set $A[m+1][1]$ to this value)\;
				\tcp{and add satellite data} 
				$A[m+1][2] \leftarrow j\sigma + i + 1$\;
				$A[m+1][3] \leftarrow a$\;
				$m \leftarrow m + 1$\;
				$\pos[a] \leftarrow m$\;
			}
		}
	}
	
	\Return $M[n][k]$\;
	\caption{The efficient algorithm from Case $1$ of \autoref{ins-distance} (on insertions).}
	\label{alg:insertions}
\end{algorithm}

{\emph{$\S$ Algorithm - initial idea.}} So, for $\ell\in [1:n]$ and $t\in [1:k]$, $M[\ell][t]$ is the minimal number of insertions needed to make $w[1:\ell]$ $t$-universal. By the explanations above, we get the following recurrence 
$M[\ell][t]=\min \{M[\ell'][t-1]+(\sigma-\Delta(\ell'+1,\ell)) \mid \ell'\leq \ell\}$. Clearly, $M[\ell][1]=\sigma - \Delta(1,\ell)$. 
Also, it is immediate to note that $M[\ell][t]\geq M[\ell''][t]$ for all $\ell\leq \ell''$. Indeed, transforming a word into a $t$-universal word can always be done with at most as many insertions as those used in transforming any of its prefixes into a $t$-universal word.

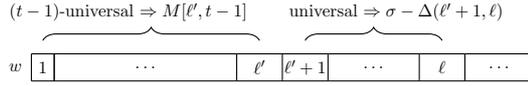
\begin{figure}[h!]
    \centering
    \begin{minipage}[c]{0.5\linewidth}
	%\centering
    	\begin{tikzpicture}[scale=0.6,every node/.style={scale=0.6}]
          % Dialectics
          
            \node[circle] (t0) at (0,0) {$w$};
            \node[wens, minimum height=6mm, minimum width=5mm, align=center, right= 0mm of t0] (t1) {$1$};
            \node[wens, minimum height=6mm, minimum width=40mm, align=center, right= 0mm of t1] (t2) {$\ldots$};
            \node[ens, minimum height=6mm, minimum width=10mm, align=center, right= 0mm of t2] (t3) {$\ell'$};
            \node[ens, minimum height=6mm, minimum width=10mm, align=center, right= 0mm of t3] (t4) {$\ell'+1$};
            \node[ns, minimum height=6mm, minimum width=20mm, align=center, right= 0mm of t4] (t5) {$\ldots$};
            \node[wns, minimum height=6mm, minimum width=10mm, align=center, right= 0mm of t5] (t6) {$\ell$};
            \node[wens, minimum height=6mm, minimum width=15mm, align=center, right= 0mm of t6] (t7) {$\ldots$};
            
            \draw[decoration={brace,raise=5pt, amplitude=5pt},decorate] (t1.north) -- node[above=10pt, xshift=-5mm]{$(t-1)$-universal $\Rightarrow M[\ell',t-1]$} (t3.north);

            \draw[decoration={brace,raise=5pt, amplitude=5pt},decorate] (t4.north) -- node[above=10pt, xshift=5mm]{universal $\Rightarrow \sigma - \Delta(\ell'+1,\ell)$} (t6.north);

        \end{tikzpicture}
    \end{minipage}
    \caption{Illustration of the formula developed for the computation of $M[\ell][t]$. }
    \label{fig:insidea}
\end{figure}

We now want to compute the elements of matrix $M$. Before this, we produce the data structures of \autoref{last_occs_ins} (and we use the notations from its framework). That is, we compute in $O(n)$ time $\last_{j\sigma+1}[a]$ and $d_{j\sigma+1}[a]=\Delta(\last_{j\sigma+1}[a],j\sigma +1)$, for all $a\in \Sigma$ and all $j\leq \frac{(n-1)}{\sigma}$. 

By \autoref{init}, we can compute the values $M[\ell][1]$, for all $\ell\in [1:n]$ in $O(n)$ time. However, a direct computation of the values $M[\ell][t]$, for $t>1$, according to the recurrence above would not be efficient. So we will analyse this recurrence further.

\smallskip

{\emph{$\S$ A useful observation.}} Assume that to transform $w[1:\ell]$ into a $t$-universal word we transform $w[1:\ell']$ into a $(t-1)$-universal word and $w[\ell'+1:\ell]$ into a $1$-universal word. The number of insertions needed to do this is $M[\ell'][t-1]+(\sigma - \Delta(\ell'+1,\ell)) $. If $w[\ell'+1]$ occurs twice in $w[\ell'+1:\ell]$, then $M[\ell'][t-1]+(\sigma - \Delta(\ell'+1,\ell)) \geq M[\ell'+1][t-1]+(\sigma - \Delta(\ell'+2,\ell))$. Thus, we can rewrite our recurrence in the following way, using the framework of \autoref{last_occs_ins}:
$M[\ell][t]=\min \{M[\ell'][t-1]+(\sigma-\Delta(\ell'+1,\ell)) \mid \ell'+1\in S_\ell\cup\{\ell+1\}\}$ (recall the definition of $S_\ell=\{\last_\ell [a]\mid a\in \letters(w[1:\ell])\}$ from \autoref{toolbox}). 
\begin{figure}[h!]
    \centering
    \begin{minipage}[c]{0.5\linewidth}
	%\centering
    	\begin{tikzpicture}[scale=0.6,every node/.style={scale=0.6}]
          % Dialectics
          
            \node[circle] (t0) at (0,0) {$w$};
            \node[wens, minimum height=6mm, minimum width=15mm, align=center, right= 0mm of t0] (t1) {$\ldots$};
            \node[wens, fill=black!25, minimum height=6mm, minimum width=5mm, align=center, right= 0mm of t1] (t2) {$ $};
            \node[wens, minimum height=6mm, minimum width=10mm, align=center, right= 0mm of t2] (t3) {$\ldots$};
            \node[wens, fill=black!25, minimum height=6mm, minimum width=5mm, align=center, right= 0mm of t3] (t4) {$ $};
            \node[wens, minimum height=6mm, minimum width=10mm, align=center, right= 0mm of t4] (t5) {$\ldots$};
            \node[wens, fill=black!25, minimum height=6mm, minimum width=5mm, align=center, right= 0mm of t5] (t7) {$ $};
            \node[wens, minimum height=6mm, minimum width=15mm, align=center, right= 0mm of t7] (t8) {$\ldots$};
            \node[wens, fill=black!25, minimum height=6mm, minimum width=5mm, align=center, right= 0mm of t8] (t9) {$ $};
            \node[wens, minimum height=6mm, minimum width=15mm, align=center, right= 0mm of t9] (t10) {$\ldots$};
            \node[wens, fill=black!40, minimum height=6mm, minimum width=10mm, align=center, right= 0mm of t10] (t11) {$\ell+1$};
            \node[wens, minimum height=6mm, minimum width=10mm, align=center, right= 0mm of t11] (t12) {$\ldots$};
            \node[minimum height=6mm, minimum width=35mm, align=center, below= 5mm of t7] (t13) {$S_l \cup \lbrace \ell +1 \rbrace$};

            \draw[->] (t13) -- (t4.south);
            \draw[->] (t13) -- (t2.south);
            \draw[->] (t13) -- (t7.south);
            \draw[->] (t13) -- (t9.south);
            \draw[dotted, thick] (t13) -- (t11.south);
            
            %, decoration={markings,mark=at position 1 with {\arrow[scale=4]{>}}}, postaction={decorate}
            
        \end{tikzpicture}
    \end{minipage}
    \caption{Only the positions $\ell'+1 \in \{\last_\ell [a]\mid a\in \letters(w[1:\ell])\} \cup \lbrace \ell +1 \rbrace=  S_\ell \cup \lbrace \ell +1 \rbrace $ are needed to compute $M[\ell][t]$ by dynamic programming. These positions are depicted here in grey. }
    \label{fig:inssl}
\end{figure}
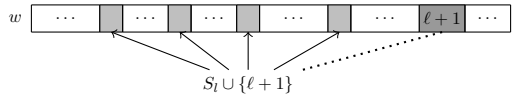

In fact, in the efficient version of our algorithm we will use a slightly weaker formula, where the minimum is computed for all elements $\ell'+1$ from a set $S'_\ell\cup\{\ell+1\}$, instead of the set $S_\ell\cup\{\ell+1\}$, where $S'_\ell$ is a superset of size at most $2\sigma$ of $S_\ell$ defined as follows. If $\ell=j\sigma+i$, for some $j\leq (n-1)/\sigma$ and $i\in [1:\sigma]$, then 
$
S'_\ell =
\left\{
	\begin{array}{ll}
		S_\ell  & \mbox{if } i=1 \\
		S_{j\sigma+1}\cup\{j\sigma+2,\ldots,j\sigma+i\} & \mbox{if } i \in [2:\sigma]
	\end{array}
\right.
$.

\smallskip

{\emph{$\S$ Algorithm - the efficient variant.}} Using the observation above, together with \autoref{speedUp}, we can compute the elements of the matrix $M$ efficiently using dynamic programming. 

So, let us consider a value $t\geq 2$. Assume that we have computed the values $M[\ell][t-1]$, for all $\ell\in [1:n]$. We now want to compute the values $M[\ell][t]$, for all $\ell\in [1:n]$.  
The main idea in doing this efficiently is to split the computation of the elements on column $M[\cdot][t]$ of the matrix $M$ in phases. In phase $j$ we compute the values $M[j\sigma+1][t],$ $M[j\sigma+2][t], \ldots, $ $M[(j+1)\sigma][t]$, for $j\leq (n-1)/\sigma$. 

We now consider some $j$, with $0\leq j\leq (n-1)/\sigma$. We want to apply \autoref{speedUp}, so we need to define the list $A$ of size $\sigma$. This is done as follows. 

We will keep an auxiliary array $\pos[\cdot]$ with $\sigma$ elements. Moreover, the element on each position $i$ of $A$, namely $A[i]$, will be accompanied by two satellite data: a position of $w$ and the letter on that position. For $a$ from $1$ to $\sigma$, if $d_{j\sigma+1}[a]=\sigma - i$ for some $i<\sigma$ then we set $A[i+1]=M[\last_{j\sigma+1}[a]-1][t-1]+i$ and $\pos[a]=i+1$; the satellite data of $A[i+1]$ is the pair $(\last_{j\sigma+1}[a],a)$. If, for some letter $a$, $\last_{j\sigma+1}[a]=n+1$ and $d_{j\sigma+1}[a]=0$ (i.e., $a$ does not occur in $w[1:j\sigma+1]$) we simply set $\pos[a]=0$. Intuitively, one can see the elements contained now in $A$ as triples: $(A[e],\last_{j\sigma+1}[a],a)$ where $A[e]=M[\last_{j\sigma+1}[a]-1][t-1]+e-1$, with $e\in[1:\sigma]$ and $a\in \Sigma$. 

Let $a_d,a_{d-1},\ldots,a_1$ be the letters of $\Sigma$ that occur in $w[1:j\sigma+1]$, ordered such that $\last_{j\sigma+1}[a_e]<\last_{j\sigma+1}[a_f]$ if and only if $e>f$. At this point, we have defined only~the~last~$d$ elements of $A$ and, for $i\in [1:d]$, the element on position $\sigma-i+1$ is $A[\sigma-i+1]=M[\last_{j\sigma+1}[a_i]-1][t-1]+(\sigma-i)$ and has the satellite data $(\last_{j\sigma+1}[a_i],a_i)$.  Also, $\pos[a_i]=\sigma-i+1$. 
The first $\sigma-d$ elements of $A$ are set to $\infty$; as convention, applying arithmetic operations to $\infty$ leaves it unchanged. We also set $m$ to $\sigma$. 

We can now define and apply a sequence of operations $o_1,\ldots,o_\sigma$ as in \autoref{speedUp}. 

{\bf An invariant:} We want to ensure that the list $A$ fulfils the following invariant properties before the execution of each operation $o_i$.
\begin{itemize}
\item For $e\in [1:d]$, the triple on position $\sigma-e+1$ of $A$ is: \\ 
$(M[\last_{j\sigma+1}[a_e]-1][t-1]+(\sigma-\Delta(\last_{j\sigma+1}[a_e], j\sigma+i)), \last_{j\sigma+1}[a_e], a_e)$.
That is, $A[\sigma-e+1]=M[\last_{j\sigma+1}[a_e]-1][t-1]+(\sigma-\Delta(\last_{j\sigma+1}[a_e], j\sigma+i))$. 
\item For $g\in [1:i-1]$, the triple on position $\sigma+g$ of $A$ is:\\
 $(M[j\sigma+g][t-1] +(\sigma-\Delta(j\sigma+g+1, j\sigma+i)), j\sigma+g+1,w[j\sigma+g+1])$.
 That is $A[\sigma+g]=M[j\sigma+g][t-1] +(\sigma-\Delta(j\sigma+g+1, j\sigma+i))$.
\item $\pos[a]$ is the position of the rightmost position $i$ storing a triple $(A[i],\ell,a)$. 
\end{itemize}
That is, the list $A$ contains all the values $M[\ell][t-1]+(\sigma-\Delta(\ell+1, j\sigma+i))$, for $\ell+1\in S_{j\sigma+1}\cup \{j\sigma+2,\ldots,j\sigma+i\}$, and $\pos[a]$ indicates the rightmost position of the list $A$ where we store a value $M[\ell][t-1]+(\sigma-\Delta(\ell+1, j\sigma+i))$ with $w[\ell+1]=a$. A consequence of this is that $A[\pos[a]]=M[\last_{j\sigma+i}[a]-1][t-1]+(\sigma-\Delta(\last_{j\sigma+i}[a], j\sigma+i))$. 
%%%% Add here something!!!

The invariant clearly holds for $i=1$. 

\smallskip

{\emph{$\S$ Algorithm - application of \autoref{speedUp}}}. In $o_i$, we extract the minimum $q$ of $A$. Then set $M[j\sigma+i][t]= \min \{q, M[j\sigma+i][t-1]+\sigma \}$. We decrement by $1$ all elements of $A$ on the positions $\pos[a]+1,\pos[a]+2, \ldots, m$, where $a=w[j\sigma+i+1]$. Then, we append to $A$ the element $M[j\sigma+i][t-1]+(\sigma-1)$, with the satellite data $(j\sigma+i+1,a)$, which implicitly increments $m$ by $1$, and set $\pos[a]=m$. 

{\bf Claim 1.} The invariant holds after operation $o_i$. %This Claim is shown in the Appendix \ref{proofs}.

{\bf Proof of Claim 1.} We now need to show that the invariant is preserved after this step. If $a=w[j\sigma+i+1]$ then the number of distinct letters occurring after each position $g>\last_{j\sigma+i}[a]$ in $w[1:j\sigma+i]$ is exactly one smaller than the number of distinct letters occurring after $g$ in $w[1:j\sigma+i+1]$. This means that $M[g-1][t-1]+ (\sigma-\Delta(g,j\sigma+i+1))$ is one smaller than $M[g-1][t-1]+ (\sigma-\Delta(g,j\sigma+i))$. Consequently, all values occurring on positions greater than $\pos[a]$ in the list $A$, which stored some values $M[g-1][t-1]+ (\sigma-\Delta(g,j\sigma+i+1))$ with $g>\last_{j\sigma+i}[a]$, should be decremented by $1$. Also, the number of distinct letters occurring after each position $g\leq \last_{j\sigma+i}[a]$ in $w[1:j\sigma+i]$ is exactly the same as number of distinct letters occurring after $g$ in $w[1:j\sigma+i+1]$. Thus, all values occurring on positions smaller or equal to $\pos[a]$ in the list $A$, which stored some values $M[g-1][t-1]+ (\sigma-\Delta(g,j\sigma+i+1))$ with $g\leq \last_{j\sigma+i}[a]$, should stay the same. So, the invariant holds for the first $\sigma+i-1$ positions of $A$. After appending $M[j\sigma+i][t-1]+(\sigma-1)$ to $A$ and incrementing $m$, then the invariant holds for the position $\sigma+i$ (which is also the last position) of $A$ too, so the invariant still holds for all positions of $A$. 

Furthermore, the only position of the $\pos$ array that needs to be updated after operation $o_i$ is $\pos[a]$, and it needs to be set to the new value of $m$. This is exactly what we do.  \qqed

%if $c\neq a$ and $\sigma <\pos[c]$ were true before executing $o_i$, and $\pos[c]$ fulfilled the invariant, then $\pos[c]$ still fulfils the invariant after executing $o_i$ because $m$ increases by $1$ and it holds that $|w[\last_{j\sigma+i+1}[c]:j\sigma+i+1]|  = |w[\last_{j\sigma+i}[c]:j\sigma+i]|+1$, so $\pos[c]=(m+1) - |w[\last_{j\sigma+i+1}[c]:j\sigma+i+1]|$. If $c\neq a$ and $\pos[c]\leq \sigma$ held before executing $o_i$, and $\pos[c]$ fulfilled the invariant, then $\pos[c]$ still fulfils the invariant after executing $o_i$ because $m$ increases by $1$ and $|\letters(w[\last_{j\sigma+1}[c] : j\sigma+1])|$ stays unchanged while $|w[j\sigma+2:j\sigma+i]|= |w[j\sigma+2:j\sigma+i+1]|-1$. Thus, $\pos[c]= (m+1) - (|\letters(w[\last_{j\sigma+1}[c] : j\sigma+1])| + |w[j\sigma+2:j\sigma+i+1]|) +1$. 
%
%For $a=w[j\sigma+i+1]$ we have that $\pos [a]$ is set to $m>\sigma$ after the execution of $o_i$. Clearly, $|w[\last_{j\sigma+i+1}[a]:j\sigma+i+1]| =1$, so the invariant is preserved for $\pos[a]$ too. 
%{\bf This concludes the proof of Claim 1.}

{\bf Claim 2.} $M[j\sigma+i][t]$ is correctly computed, for all $i\in [1:\sigma]$. 

{\bf Proof of Claim 2.} According to the invariant, before executing operation $o_i$, $A$ contains the values $M[\ell][t-1]+(\sigma-\Delta(\ell+1, j\sigma+i))$, for $\ell+1\in S_{j\sigma+1}$, and  $M[j\sigma+g][t-1] +(\sigma-\Delta(j\sigma+g+1, j\sigma+i))$, for $g\in [1:i-1]$. As $S'_{j\sigma+i}=S_{j\sigma+1} \cup \{j\sigma+g+1\mid g\in[1:i-1]\}$ is a superset of size at most $2\sigma$ of $S_{j\sigma+i}$, we obtain that $M[j\sigma+i][t]$ is correctly computed as the minimum between the smallest value in $A$ and $M[j\sigma+i][t-1]+\sigma$. \qqed 
%{\bf This concludes the proof of Claim 2.}

\smallskip

{\emph{$\S$ Algorithm - the result of applying \autoref{speedUp}}}.  After executing the $\sigma$ operations $o_1,\ldots, o_\sigma$, we have computed the values $M[j\sigma+1][t],$ $M[j\sigma+2][t], \ldots, $ $M[(j+1)\sigma][t]$ correctly. We can move on to phase $j+1$ and repeat this process. 

\smallskip

{\emph{$\S$ The result and complexity.}} The minimal number of insertions needed to make $w$ $k$-universal is, according to the observations we made, correctly computed as $M[n][k]$. 

 By \autoref{speedUp}, computing $M[j\sigma+1][t],$ $M[j\sigma+2][t], \ldots, $ $M[(j+1)\sigma][t]$ takes $O(\sigma)$ for each $j$. Overall, computing the entire column $M[\cdot][t]$ takes $O(n)$ time. We do this for all $t\leq k$, so we use $O(nk)$ time in total to compute all elements of $M$. 
This concludes Case 1. 

\smallskip

{\bf Case 2.} Assume $k>n$. This case can also be solved using the algorithm from Case~1. 

\smallskip

\emph{$\S$ The reduction.} Let $w'$ be a $k$-universal word obtained from $w$ by a minimal number of insertions, then, by \autoref{theodlt}, there exist the $1$-universal words $w'_1,\ldots,w'_k$, and the word $w'_{k+1}$ such that $w'=w'_1\cdots w'_kw'_{k+1}$ and $\letters(w'_{k+1})$ is a strict subset of $\Sigma$. That is, $w'_1, \ldots, w'_k$ are the arches of $w$. As $k>n$, there will be at most $n$ arches $w'_{i_1},\ldots,w'_{i_n}$, with $i_j\in [1:k]$ and $i_{j}< i_{j+1}$ for all $j$, which contain original letters of $w$ (i.e., all other letters of $w'$ were inserted). That is, $w$ is a subsequence of $w''=w'_{i_1}\ldots w'_{i_n}$, and $w''$ is an $n$-universal word obtained from $w$ by a minimal number of insertions. So, to compute the minimal number of insertions needed to transform $w$ into a $k$-universal word we can first compute the minimal number
of insertions needed to transform $w$ into an $n$-universal word $w''$, and then add $(k-n)\sigma$, the minimal number of insertions needed to transform $w''$ into $w'$ (i.e., the length
of the missing $1$-universal factors $w'_i$ from the decomposition of $w'$). Using the algorithm described in the case $k\leq n$, we can compute in $O(n^2)$ time the minimal number of insertions needed to transform $w$ into an $n$-universal word (which is clearly at most $n\sigma$), and then we can compute in $T(n,\sigma,k)$ time the value $(k-n)\sigma$. By this, Case 2 is now complete.
This also concludes the proof of the theorem.
\end{proof}

\begin{remark}
If $k$ is in $O(c^n)$, for some constant $c$, then $T(n,\sigma,k)\in O(n^2)$ so the algorithm from Case 2 of \autoref{ins-distance} will run in $O(n^2)$~time. 
\end{remark}

%%%%%%%%%%%%%%%%%%%%%%%%%%%%

\subsection{Deletions}\label{del}

%%%%%%%%%%%%%%%%%%%%%%%%%%%%

 \begin{theorem} \label{deletions}
Let $w$ be a word, with $|w|=n$, $\letters(w)=\Sigma$, and $\Sigma = \{1,2,\ldots,\sigma\}$. 
Let $k$ be an integer with $k \leq \iota(w)\leq n/\sigma$. We can compute in $O(nk)$ time the minimal number of deletions needed to obtain a word of universality index~$k$ (w.r.t. $\Sigma$) from $w$.
\end{theorem}

\begin{proof}
The proof is again structured in such a way that the idea of the solution, as well as the actual computation steps, and the arguments supporting their correctness are clearly marked. See also \autoref{alg:deletions}.

\smallskip

{\em $\S$ General approach.} The case $k=0$ is trivial. We just need to count how many times each letter occurs and then remove the letter that occurs the least number of times. This takes $O(n)$ time. So let us assume that $k>0$. To simplify the presentation, we will call a word $u$ a weak-$p$-universal word if $u$ is $p$-universal and by deleting the last letter of $u$ we obtain a word of universality index $p-1$.

Assume that $w'$ is a $k$-universal word that can be obtained by applying the sequence of deletions of minimal length to $w$. 
Clearly, $w'$ is a subsequence of $w$, 
and, by the decomposition defined in the context of \autoref{decomp}, 
there exist the arches $w'_1, \ldots, w'_k$, all of universality index exactly $1$, 
and $w'_{k+1}$, of universality index $0$, 
such that $w' = w'_1\cdots w'_k w'_{k+1}$. Moreover, each arch $w'_i$, with $i\in[1:k]$, is a weak-$1$-universal word. 

It follows that actually each of the words $w'_i$ is a subsequence of $w$ too. 
So we will try to identify the factors  $w[i_{j-1}+1:i_{j}]$, with $j\in [1:k+1]$ and $i_0=0$, from 
which $w'_j$ is obtained by deletions, with the important condition that the last letter of $w[i_{j-1}+1:i_{j}]$ is not deleted (so the last letters of $w'_j$ and $w[i_{j-1}+1:i_{j}]$ coincide).

Note that, to obtain the word $w'_1\cdots w'_k w'_{k+1}$ from $w$ with a minimal number of deletions, we need to find a position $i_k$ of $w$ such that $w'_1\cdots w'_k $ is obtained from $w[1:i_{k}]$ and $w'_{k+1}$ from $w[i_{k}+1:n]$, with the restrictions that $w[i_k]$ is not deleted and the overall number of deletions made in this process is the smallest (over all possible choices of the position $i_k$). If we now have $i_k$, we can then search for $i_{k-1}$: we need to find a position of $w$ such that $w'_1\cdots w'_{k-1} $ is obtained from $w[1:i_{k-1}]$ and $w'_{k}$ from $w[i_{k-1}+1:i_k]$, with the similar restrictions that $w[i_{k-1}]$ is not deleted and the overall number of deletions made in this process is the smallest (over all possible choices of the position $i_k$). And then we continue with $i_{k-2}$, $i_{k-3}$, and so on. 

Thus, it seems natural to consider and solve the following type of subproblems: what is the minimal number of deletions we need to apply to transform a prefix $w[1:i]$ into a weak-$p$-universal word $v$, without deleting $w[i]$. 

\smallskip

\SetKw{KwAnd}{and}
\begin{algorithm}[h!]
	\SetKwInOut{Input}{Input}
	\SetKwInOut{Output}{Output}
	\Input{word $w$, alphabet $\Sigma$, int $k$}
	\Output{minimal number of deletions}
	\BlankLine
	
	\tcp{initialization}
	int $n \leftarrow \lvert w \rvert$; int $\sigma \leftarrow \lvert \Sigma \rvert$\;
	int $N[n][k] = \infty$\;
	\BlankLine
	
	\tcp{initialise first column of $N$}
	\For{$i = 1$ \KwTo $n$}{
		\If{$\Delta(1,i)=\sigma$ (i.e., $w[1:i]$ is $1$-universal)}{
			$N[i][1] \leftarrow \freq[i]-1$\;
		}
	}
	\BlankLine
	
	\tcp{efficient variant}
	\For{$p = 2$ \KwTo $k$}{
		int $M'[n]$\;
		\For{$i=1$ \KwTo $(p-1)\sigma$}{
			$M'[i] \leftarrow \infty$\;
		}
	
		\For{$i=(p-1)\sigma$ \KwTo $n$}{
			int $l \leftarrow \last_{i-1}[w[i]]$\;
			\uIf{$l = n+1$}{
				$M'[i] \leftarrow \infty$\;
			} \Else {
				int $r \leftarrow \RMQ_{p-1}(l+1, i-1)$\;
				$M'[i] \leftarrow 1+\min\{M'[l], N[r][p-1]\}$\;
			}
		}
		\BlankLine
	
		\tcp{compute $N[\cdot][p]$ using $M'$}
		\For{$i = 1$ \KwTo $n$}{
			int $j \leftarrow \univ[i]$\;
			\uIf{$j = 0$}{
				$N[i][p] \leftarrow \infty$\;
			}\Else{
				int $t \leftarrow \last_{j-1}[w[i]]$\;
				\uIf{$t = n+1$}{
					$N[i][p] \leftarrow \infty$\;
				}\Else{
					int $r \leftarrow \RMQ_{p-1}(t+1, j-1)$\;
					$N[i][p] \leftarrow \min\{M'[t] + \freq[i] - \freq[t] - 1, N[r][p-1] + \freq[i] - \freq[t] - 1 \}$\;
				}
			}
		}
	}
	
	\Return $\min\{N[i][k] + T[i] \mid 1 \leq i \leq n\}$\;
	
	\caption{The efficient algorithm from \autoref{deletions} (on deletions).}
	\label{alg:deletions}
\end{algorithm}

{\em $\S$ Algorithm - initial idea.} We first compute all the data structures defined in \autoref{help1}.

We define the $n\times k$ matrix $N$, 
where $N[i][p]$ is the minimal number of deletions we need to apply to $w[1:i]$, without deleting $w[i]$, to obtain a weak-$p$-universal word $v$ from it (for $1 \leq i \leq n$ and $1 \leq p \leq k$). If $w[1:i]$ is not $p$-universal, $N[i][p]$ will be set to $\infty$. 

To compute this matrix, we note that if $N[i][1]\neq \infty$, then $N[i][1]=\freq[i]-1$. Indeed, in order to produce from $w[1:i]$ a weak $1$-universal word, which fulfils the conditions stated above, we have to delete all occurrences of the letter $w[i]$ except for the one on position $i$. 

In general, the minimal number of deletions needed to transform a factor $w[i:j]$ of universality index at least $1$ (so with $i\leq \univ[j]$) into a weak-$1$-universal word $v$, such that $w[j]$ is not deleted, is $|w[i:j]|_{w[j]}-1$.

Accordingly, we can define the elements of the matrix $N$ as $N[i][p] = \min \{N[i'][p-1] + |w[i'+1:i]|_{w_i}\mid i'\in [1:\univ[i]-1]\}$, for $i\in [1:n], p\in [2:k]$.  

A straightforward implementation of the above formula is not efficient, so we will explore alternative and more efficient ways to compute $N[i][p]$. 

\smallskip

{\em $\S$ Algorithm - an efficient implementation.} We will show how the elements of the column $N[\cdot][p]$ can be computed efficiently for a given $p\geq 2$, 
assuming that we have computed them for $N[\cdot][p-1]$. 

Firstly, we define the data structures of \autoref{RMQ} allowing us to answer range minimum queries for the column $N[\cdot][p-1]$ of $N$, denoted by $\RMQ_{p-1}$ in the following. 
Note that the query $\RMQ_{p-1}(j,i)$ returns, for some $j<i$, the position $k$ with $k\in [j:i]$, 
such that $w[1:k]$ is the prefix of $w$, ending between $j$ and $i$, which can be transformed into a weak-$(p-1)$-universal word with the least number of deletions (among all other prefixes ending between $i$ and $j$), such that its last letter $w[k]$ is not deleted.

We define an auxiliary array $M'[\cdot]$ with $n$ elements, where if $i\geq (p-1)\sigma$, we have $M'[i]=\min\{N[j][p-1] + |w[j+1:i]|_{w[i]} \mid j<i\}$, and if $i<(p-1)\sigma$, we have $M'[i]=\infty$ (that is, a large enough value). 

Intuitively, $M'[i]$ is the minimal number of deletions we need to make in $w[1:i]$ in order to obtain a word $v'v''$, where $v'$ is a weak-$(p-1)$-universal word and $v''$ is a word with universality index $0$ which contains no occurrence of $w[i]$.

We observe that the values $M'[i]$ can be computed as follows. Firstly, we set $M'[i]=\infty$ if $i<(p-1)\sigma$. Then, for $i\geq (p-1)\sigma$, let $\ell= \last_{i-1}[w[i]]$.

\smallskip

{\bf Claim 1.} We claim that $M'[i]= 1 + \min \{M'[\ell], N[\RMQ_{p-1}(\ell+1,i-1)][p-1]\}$, if $\last_{i-1}[w[i]]\neq n+1$, or $M'[i]=\infty$ otherwise. 

{\bf Proof of Claim 1.} The case when $w[i]$ does not occur in $w[1:i-1]$ is trivial. The first part of the claim holds because the minimal number of deletions we need to make in $w[1:i]$ in order to obtain a word $v'v''$, where $v'$ is a weak-$(p-1)$-universal word and $v''$ is a word with universality index $0$ which contains no occurrence of $w[i]$ is:
\begin{itemize}
\item either the minimal number of deletions we need to apply to $w[1:\ell]$ in order to obtain a word $v_0'v_0''$, where $v_0'$ is a weak-$(p-1)$-universal word and $v''$ is a word with universality index $0$ which contains no occurrence of $w[\ell]=w[i]$, and then remove the last occurrence of $w[i]$ from $w[1:i]$, 
\item or the minimal number of deletions we need to apply to $w[1:i]$ in order to obtain a word $v'v''$, where $v'$ is a weak-$(p-1)$-universal word obtained from a prefix $w[1:j]$, with $j>\ell$, by deleting some letters, but not $w[j]$, and $v''$ is a word with universality index $0$ which contains no occurrence of $w[i]$, obtained from $w[j+1:i]$ by the single deletion of $w[i]$. 
\end{itemize}
This {\bf concludes the proof of Claim 1.}  

\smallskip 

Computing $M'$ takes linear time. Now, we can use the array $M'$ to compute the values stored in $N$. 
Assume that we want to compute $N[i][p]$. 

Let $j=\univ[i]$. If $j=0$, then set $N[i][p]=\infty$. Assume in the following that $j\neq 0$.  
Clearly, $i \in L_j$. 
Therefore, the structures we computed with \autoref{help1} provide the value $t= \last_{j-1}[w[i]]$. If $t=n+1$ (i.e., $w[i]$ does not occur in $w[1:j-1]$), and because $p\geq 2$, we set $N[i][p]=\infty$ (we cannot transform the word $w[1:i]$ into a $p$-universal word by deletions when it only contains one occurrence of $w[i]$). 
If $t<n+1$, let $r=\RMQ_{p-1}(t+1,j-1)$. 

\smallskip 

{\bf Claim 2.} We claim that the following holds $N[i][p]=\min \{M'[t]+|w[t+1:i]|_{w[i]}-1,$ $ N[r][p-1]+|w[r+1:i]|_{w[i]} -1\}$. 

{\bf Proof of Claim 2.} Indeed, this is true because in order to transform (with a minimal number of deletions) $w[1:i]$ into a weak-$p$-universal word without deleting $w[i]$ we can 
\begin{itemize}
\item either transform (with a minimal number of deletions) a prefix $w[1:t']$ of $w[1:t]$ into a weak-$(p-1)$-universal word and remove all the occurrences of $w[i]$ from $w[t'+1:t]$, and, then, all occurrences of $w[i]$ from $w[t+1:i]$, except $w[i]$, 
\item  or we transform (again with a minimal number of deletions) a prefix $w[1:t']$ of $w[1:j]$, with $t'>t$ and $j=\univ(i)$, into a weak-$p-1$-universal word, and then remove all the occurrences of $w[i]$ from $w[t'+1:i]$, except $w[i]$.
\end{itemize}
This {\bf concludes the proof of Claim 2.} 

\smallskip

Let us now note that $M'[t]+|w[t+1:i]|_{w[i]}-1=M'[t]+\freq[i]-\freq[t]-1$ (because $w[i]=w[t]$). 
Also, $N[r][p-1]+|w[r+1:i]|_{w[i]} -1 = N[r][p-1]+|w[t+1:i]|_{w[i]} -1 = N[r][p-1]+\freq[i]-\freq[t]-1$ because $w[i]$ does not occur in $w[t+1:r]$ (as $w[i]$ does not occur between $w[t+1:j-1]$). 

This gives us a way to compute each $N[i][p]$ in constant time (once $M'[\cdot]$ was computed).

Thus, computing the entire column $N[\cdot][p]$ takes overall linear time $O(n)$ (including here the computation of the array $M'$). 

Consequently, in total, we can compute the elements of the matrix $N$ in $O(nk)$ time. 

\smallskip

{\em $\S$ Collecting the results.} We are not done yet, as it is not clear which is the minimal number of deletions we need in order to transform $w$ into a $k$-universal word. 

Recall that \autoref{help1} also computes the array $T[\cdot]$ with $T[i]=\min \{|w[i+1:n]|_{a} \mid a \in \Sigma \}$, for $i\in [0:n]$.

Now, the minimal number of deletions we need in order to transform $w$ into a $k$-universal word is clearly $\min \{N[i][k] + T[i] \mid 1 \leq i \leq n\}$: we check which is the minimal number of deletions we need in order to both transform a prefix $w[1:i]$ into a weak-$p$-universal word, without deleting $w[i]$, and the word $w[i+1:n]$ into a word with universality index $0$.

\smallskip

{\em $\S$ Complexity.} According to the above, the answer returned by our algorithm can be computed in $O(n)$ time, after the matrix $N$ was computed. So, overall, the time complexity of the algorithm is $O(nk)$.  

The correctness of this approach follows from the observations we made during the explanation of the algorithm. So, the statement follows. 
\end{proof}
%The idea of this proof is the following. Assume that $w'$ is a word of universality index $k$ obtained via the sequence of deletions of minimal length from $w$. Clearly, $w'$ is a subsequence of $w$, and, by the decomposition defined in Theorem \ref{decomp}, there exist $w'_1,\ldots,w'_k$, all of universality index exactly $1$, and $w'_{k+1}$, of universality index $0$, such that $w'=w'_1\cdots w'_k w'_{k+1}$. It follows that each of the words $w'_i$ is a subsequence of $w$ too. So we will try to identify each subsequence $w'_1\cdots w'_p$ for $p\leq k$ and the shortest factor $w[1:i]$ from which it is obtained. To this end, we define the matrix $N$, where $N[i][p]$ is the minimal number of deletions we need to apply to $w[1:i]$, without deleting $w[i]$, to obtain a word $v$ from it, with $\iota(v)=p$ and $\iota(v[1:|v|-1])=p-1$~(for 
%$i\in [1:n]$ and $p\in [1:k]$). If $\iota(w[1:i])\geq 1$, then $N[i][1]=|w[1:i]|_{w[i]}-1$, as~we~have~to~delete all occurrences of $w[i]$ from $w[1:i]$, except the one on position $i$. Then, $N[i][p]=\min\{N[j][p-1] + |w[j+1:i]|_{w[i]}-1\mid j<i$ such that $\iota (w[j+1:i])\geq 1\}$. This gives a dynamic programming algorithm for computing $N$. Using additional data structures based on Lemma~\ref{RMQ}, we can compute the elements of $N$ in $O(nk)$ time. To show the statement, we need to return $\min \{N[i,k] + T[i] \mid 1 \leq i \leq n\}$, using the array $T$ computed in Lemma \ref{help1}.

\input{substitutions}

\section{Extensions, Conclusions and Future Work}
%%%%%%%%%%%%%%%%%%%%%
In this paper, we presented a series of algorithms computing the minimal number of edit operations one needs to apply to a word $w$ in order to reach $k$-subsequence universality. In fact (see \autoref{construction}, \autoref{proofs}), one can extend our algorithms and, using additional $O(k|\letters(w)|)$ time, we can effectively construct a $k$-universal word which is closest to $w$, with respect to the edit distance. All our algorithms can be implemented in linear space (see \autoref{space} and \autoref{construction}, \autoref{proofs}) using a technique called {\em Hirschberg's trick}~\cite{Hirschberg75}.

The algorithms we presented work in a general setting: the processed words are over an integer alphabet. 
It seems natural to ask whether we can devise faster solutions for inputs over a binary alphabet. To this end, we can show (using a greedy strategy) that the result of
\autoref{ins-distance} can be improved: 
%\begin{theoremEnd}{theorem}\label{distance-binary}
%Let $w$ be a word, with $|w| = n$ and $alph(w)=\{\ta,\tb\}$. Let $k\geq \iota(w)$ be a positive integer.
we can compute in $O(|w|)$ time the minimal number of insertions needed to obtain a $k$-universal~word from a given binary word $w$ (\autoref{distance-binary}, \autoref{proofs}). 
It is open whether this holds for substitutions or deletions. 
Moreover, it is open whether we can extend this result to the case of constant size input alphabets, and obtain that we can compute in $O(nf(\sigma))$ time the minimal number of insertions (substitutions, or, respectively, deletions) needed to apply to $w$ to obtain a $k$-universal~word, where $\sigma=|\letters(w)|$ and $f$ is some function. Note that the time complexities of our algorithms from \autoref{editDist} do not depend on $\sigma$, so probably a new approach would be needed.

%Our restriction that we want to obtain from the input word $w$ a new word with universality index $k$ {\em with respect to the alphabet $\letters(w)$} is crucial. If we drop this restriction, and we  want instead to obtain a word $w'$ from $w$ whose universality index is still $k$ but, this time, with respect to $\letters(w')$, then the opening remarks we made in Section \ref{editDist} do not hold anymore. For instance the word $\tb\ta\tn\ta\tn\ta\tn$ can be transformed into a $3$-universal word  w.r.t. $\{\tb,\ta,\tn\}$ (e.g., $\tb\ta\tn\tb\ta\tn\ta\tn\tb$) with $2$ insertions, but also with only one deletion, e.g., into the word $\ta\tn\ta\tn\ta\tn$ which is $3$-universal w.r.t. $\{\ta,\tn\}$. Investigating how the edit-distance to $k$-universality can be computed in this setting seems an interesting extension of this work to us.
%
The main open direction of research remains computing the edit distance from a given word $w$ to the set of words whose $k$-spectrum equals (or includes) a given set $S$. Depending on how $S$ is specified (e.g., as the $k$-spectrum of some word, or represented in other compact form), it is expected that novel techniques will need to be developed to solve such problems.

\newpage

\bibliography{scatteredUniversalityArxiv}

\newpage

\renewcommand{\thesection}{\Alph{section}}

\begin{appendices}

\section{Appendix: Universality Alphabet} \label{alphabet}
Our restriction that we want to obtain from the input word $w$ a new word with universality index $k$ {\em with respect to the alphabet $\letters(w)$} is crucial for our approach and results. If we drop this restriction, and instead we want to obtain a word $w'$ from $w$ whose universality index is still $k$ but, this time, with respect to $\letters(w')$, then the opening remarks made in \autoref{editDist} do not hold anymore. For instance the word $\tb\ta\tn\ta\tn\ta\tn$ can be transformed into a $3$-universal word  w.r.t. $\{\tb,\ta,\tn\}$ (e.g., $\tb\ta\tn\tb\ta\tn\ta\tn\tb$) with $2$ insertions, but also with only one deletion, e.g., into the word $\ta\tn\ta\tn\ta\tn$ which is $3$-universal w.r.t. $\{\ta,\tn\}$. Investigating how the edit-distance to $k$-universality can be computed in this setting  seems a possible extension of this work to us. However, it seems more natural to define the alphabet used as reference for universality as the alphabet of the input, rather than computing dynamically it during the algorithm. 

\section{Appendix: Additional Results and Proofs}\label{proofs}

\begin{remark}\label{space}
Assume $k\leq n$. In the proofs of \Cref{ins-distance,edit-subs} and, respectively, in the proof of \autoref{deletions} the main part is computing the matrices $M$ and respectively $N$. This requires $O(nk)$ time and space. However, as computing the columns $M[\cdot][p]$ and $N[\cdot][p]$ only requires knowing the values on columns $M[\cdot][p-1]$ and, respectively, $N[\cdot][p-1]$, we can reduce the space consumption to $O(n)$. So, if we are only interested in computing the minimal number of insertions, deletions, or substitutions required to transform $w$ into a $k$-universal word, $O(n)$ space and $O(nk)$ time are enough.

The case $k>n$ is only relevant when we want to compute the minimal number of insertions required to transform $w$ into a $k$-universal word. As explained in the proof of \autoref{ins-distance}, this can be computed in $O(n^2)$ time to which we add the time needed to compute $\sigma(k-n)$. Similarly, by the observations made above, we can compute the minimal number of insertions required to transform $w$ into a $k$-universal word in $O(n)$ space (the computation of the matrix $M$) to which we add the space needed to compute $\sigma(k-n)$. \qed
\end{remark}

\begin{theorem}\label{construction}
Let $w$ be a word, with $|w|=n$, $\letters(w)=\Sigma$, and $\Sigma = \{1,2,\ldots,\sigma\}$. Let $k\neq \iota(w)$ be an integer. We can construct one of the $k$-universal words which are closest to~$w$ w.r.t. edit distance in $O(kn)$  time, if $k\leq n$, and $O(n^2+k\sigma)$ time, otherwise. The space needed for this construction is $O(n+k\sigma)$
\end{theorem}
\begin{proof}

{\em $\S$ The initial algorithm.} We first explain how one of the $k$-universal words which are closest to~$w$ w.r.t. edit distance can be constructed in $O(kn)$ time, without fulfilling the space complexity restriction. We will split the discussion in two cases. 

{\bf Case 1.} $k> \iota(w)$. In this case, we only need insertions to produce one of the $k$-universal words which are closest to $w$ w.r.t. edit distance. So we will use the algorithm in \autoref{ins-distance} to construct this word. We assume that we use the same notations as in the proof of the respective theorem. In the referenced algorithm we compute, for each position $\ell$ of $w$ and each $t\leq \min\{n,k\}$ a value $j_\ell$ such that 
$$M[\ell][t]=M[j_\ell][t-1]+(\sigma-\Delta(j_\ell+1,\ell) = \min \{M[\ell'][t-1]+(\sigma-\Delta(\ell'+1,\ell)) \mid \ell'\leq \ell\}.$$
We define $Sol[\ell][t]\gets j_{\ell}$. 

Let us assume first $k\leq n$. Define the sequence $i_{k-1}=Sol[n][k]$ and, for $j\in [2,k-1]$, $i_{j-1}=Sol[i_j][j]$. Let $i_0=0$ and $i_k=n$. It is not hard to see that for the decomposition $w=v_1\cdots v_k$, where, for $j\in [1:k-1]$, $v_j=w[i_{j-1}+1:i_j]$, the following holds: $\sum_{j\in [1:k]}(\sigma-\Delta(i_{j-1}+1,i_j))=M[n][k]$. In other words, if we insert the minimal number of letters in each of the words $v_1,\ldots,v_k$ such that they become universal, then we obtain one of the $k$-universal words which are closest to $w$ w.r.t. the edit distance. 

Clearly, the sequence of words $v_1=[1:i_1]$, $v_2=[i_1+1:i_2], \ldots $,$v_k=[i_{k-1}+1:i_k]$ can be computed in linear time $O(n)$ once we have the matrix $Sol$. 

Now, to compute a $k$-universal word obtained from $w$ by making each of the words $v_1,\ldots, v_k$ universal with a minimal number of insertions, we do the following. For each $i\in [1:k]$, by traversing the word $v_i$ left to right we can identify in $O(|v_i|+\sigma)$ the subset $V_i$ of $\Sigma$ containing the letters which do not occur in $v_i$ (e.g., using a counting vector like in the proof of \autoref{init}). We produce a word $u_i$ from $v_i$ by appending the letters from $V_i$ at the end of $v_i$; this takes $O(|v_i|+\sigma)$ time. Then we concatenate the words $u_1,\ldots, u_k$ and obtain a word $u$ which is $k$-universal and the number of insertions needed to obtain $u$ from $w$ is $M[n][k]$. The total time needed to produce $u$ is $O(k\sigma)$.

Let us now assume that $k> n$. Just like above, we obtain an $n$-universal word $u$ from $w$. Then we concatenate at the end of $u$ the word $(1\cdot 2\cdots \cdot \sigma)^{k-n}$. In this way, we obtain a word $u'$ which is $k$-universal and the number of insertions needed to obtain $u$ from $w$ is minimal, i.e., $M[n][k]+(k-n)\sigma$. The total time needed to produce $u'$ is $O(n^2+k\sigma)$.

This concludes the analysis of Case 1. 

\smallskip

{\bf Case 2.} $k< \iota(w)$. In this case, we only need deletions to produce one of the $k$-universal words which are closest to $w$ w.r.t. edit distance. So we will use the algorithm in \autoref{deletions} to construct this word. We assume that we use the same notations as in the proof of the respective theorem. In the algorithm described in the proof of \autoref{deletions} we compute, for $i\in [1:n], p\in [2:k]$, a value $j_i$ such that 
$$N[i][p] =N[j_i][p-1] + |w[j_i+1:i]|_{w_i} = \min \{N[i'][p-1] + |w[i'+1:i]|_{w_i}\mid i'\in [1:\univ[i]-1]\}.$$

We define $Sol[i][p]\gets j_{i}$. We have $Sol[i][1]=0$.

Finally, to return the minimal number of deletions needed to transform $w$ into a $k$-universal word, we compute $m=\argmin\{N[i][k] + T[i] \mid 1 \leq i \leq n\}$.

We now define the sequence $i_k=m$, $i_{j-1}=Sol[i_j][j]$, with $j\geq 2$, and $i_0=1$. Let $i_{k+1}=n$. It is not hard to see that for the decomposition $w=v_1\cdots v_kv_{k+1}$, where, for $j\in [1:k+1]$ and $v_j=w[i_{j-1}+1:i_j]$, the following holds:  $\sum_{i\in [1:k]} (|v_i|_{v_i[|v_i|]}+1) + |v_k|_{w[m]} = \min\{N[i][k] + T[i] \mid 1 \leq i \leq n\}$. 

Clearly, the sequence of words $v_1, \ldots, v_k, v_{k+1}$ can be computed in linear time $O(n)$ once we have the matrix $Sol$. To compute a $k$-universal word obtained from $w$ by making each of the words $v_1,\ldots, v_k$ universal with a minimal number of deletions, we just remove from $v_i$ all the occurrences of their last letter (i.e., $v_i[|v_i|]$), except the rightmost one. This takes $O(n)$ time, and we can output the word obtained by this procedure as one of the $k$-universal words which are closest w.r.t. edit distance to $w$.

\smallskip

{\em $\S $ A space efficient implementation.} We will only discuss in detail how the case when $n\geq k>\iota(w)$ is implemented, as all the other cases can be approached in exactly the same manner. 

We know by \autoref{space} that the matrix $M$ can be computed in linear space $O(n)$. However, it is unclear how we can compute the sequence $i_0,i_1,\ldots, i_k$ in linear space. In particular, in the approach described above we explicitly need to compute and store all the elements in the matrix. 

Fortunately, there exists a standard way to deal with this problem (known as {\em Hirschberg's trick} \cite{Hirschberg75}). 

We will need to define more formally the problem that we want to solve in this framework. 

We want to solve the problem ${\mathcal P}(w)$ which requires, for the input word $w$ of length $n$, to compute the smallest number $m$ of insertions needed to transform $w$ into a $k$-universal word {\bf and} the sequence of positions $i_0=1,i_1,\ldots, i_{k-1}, i_k=n$ such that we can transform each of the words $w[i_{j-1}:i_j]$, with $j\in [1:k]$, into a universal word, using in total exactly $m$ insertions. 

The solution of ${\mathcal P}(w)$ is the following.

Firstly, we will compute the value $m$ as described in the proof of \autoref{ins-distance}, using only linear space as described in \autoref{space}, but, alongside $m$, we will also compute the value $i_{\lfloor k/2 \rfloor}$. 

This can be done by computing (the columns of) an additional matrix  $H$, simultaneously with (the columns of) the matrix $M$. Recall that computing $M[\ell][t]$ is based on identifying a position $j_\ell$ such that 
$M[\ell][t]=M[j_\ell][t-1]+(\sigma-\Delta(j_\ell+1,\ell).$ Then $H[\ell][t]$ is defined as follows: 
$
H[\ell][t] =
\left\{
	\begin{array}{ll}
		\infty  & \mbox{if } t<\lfloor k/2 \rfloor,\\
		H[j_\ell][t-1] & \mbox{if } t> \lfloor k/2 \rfloor, \\
		j & \mbox {if } t=\lfloor k/2 \rfloor.
	\end{array}
\right.
$

Intuitively, given that $M[\ell][t]$ is the minimal number of insertions needed to transform $w[1:\ell]$ into a $t$-universal word, then $H[\ell][t]$ is the ending position of the prefix of $w[1:\ell]$ which was transformed in a $\lfloor k/2 \rfloor$-universal word when the respective sequence of length $M[\ell][t]$ of insertions is applied to $w[1:\ell]$ .

Clearly, to compute the elements of the column $H[\cdot][t]$ of the matrix $H$ we only need column $H[\cdot][t-1]$ (and the columns $M[\cdot][t]$ and $M[\cdot][t-1]$ of matrix $M$). So, we can compute $H[n][k]$ in linear space and $O(nk)$ time. Also, it is not hard to see that $H[n][k]$ is exactly the value $i_{\lfloor k/2 \rfloor}$. (In fact, there may be more solutions to our problem ${\mathcal P}(w,k)$, so $H[n][k]$ corresponds to the position $i_{\lfloor k/2 \rfloor}$ in one of these solutions)

To compute the rest of the values $i_0,\ldots,i_{\lfloor k/2 \rfloor-1},i_{\lfloor k/2 \rfloor+1},\ldots,i_k $ we proceed in a divide and conquer manner. We solve ${\mathcal P}(w[1:i_{\lfloor k/2\rfloor}],{\lfloor k/2\rfloor})$ and obtain 
the sequence $i_0,\ldots,i_{\lfloor k/2 \rfloor-1},i_{\lfloor k/2 \rfloor}$, and then we solve ${\mathcal P}(w[i_{\lfloor k/2\rfloor}+1:n],{\lceil k/2\rceil})$ and we obtain the sequence $i_{\lfloor k/2 \rfloor},i_{\lfloor k/2 \rfloor+1},\ldots,i_k$. Then we simply return $i_0,\ldots,i_{\lfloor k/2 \rfloor-1},i_{\lfloor k/2\rfloor},i_{\lfloor k/2 \rfloor+1},\ldots,i_k $ and the value $M[n][k]$ as a solution to ${\mathcal P}(w,k)$.

The correctness of the algorithm is based on the following simple remark: if $w[1:n]$ can be transformed into a $k$-universal word by a sequence of $m$ insertions, such that in the respective sequence of insertions transforms $w[1:i_{\lfloor k/2\rfloor}]$ into a ${\lfloor k/2\rfloor}$-universal word, then the following hold:
\begin{itemize}
\item  $w[1:i_{\lfloor k/2\rfloor}]$ can be transformed into a ${\lfloor k/2\rfloor}$-universal word by $p$ insertions;
\item  $w[i_{\lfloor k/2\rfloor}+1:n]$ can be transformed into a ${\lceil k/2\rceil}$-universal word by $m-p$ insertions.
\end{itemize} 
Thus, solving ${\mathcal P}(w,k)$ can be reduced to computing $i_{\lfloor k/2\rfloor} $ and solving recursively ${\mathcal P}(w[1:i_{\lfloor k/2\rfloor}],{\lfloor k/2\rfloor})$ and ${\mathcal P}(w[i_{\lfloor k/2\rfloor}+1:n],{\lceil k/2\rceil})$.

The time complexity $T(n,k)$ of solving ${\mathcal P}(w,k)$ is then $T(n,k)=O(nk) + T(i_{\lfloor k/2\rfloor},\lfloor k/2\rfloor) + T(n-i_{\lfloor k/2\rfloor}, {\lceil k/2\rceil})$. It is easy to show that $T(n,k)\in O(nk)$. 

Assume that the algorithm computing $H[n][k]$ and $M[n][k]$ needs $cn$ space for some constant $c$. 

One can show by induction on $n+k$ that the space complexity $S(n,k)$ of solving ${\mathcal P}(w,k)$ is upper-bounded by $dn$, where $d\geq c$ is a constant, $n=|w|$, and $k\leq n$. The case $n+k=2$ is trivial. Assume that this is true for $n+k\leq r$. We show that it is true for $n+k=r+1$. To solve ${\mathcal P}(w,k)$ we first compute the values $M[n][k]$ and $H[n][k]$ in $cn\leq dn$ space, for some constant $c$. The main observation we make here is that the space we used in this computation can then be reused. Then we solve ${\mathcal P}(w[1:i_{\lfloor k/2\rfloor}],\lfloor k/2\rfloor)$ using $di_{\lfloor k/2\rfloor}\leq dn$ space (which is actually reused space). Finally, we solve ${\mathcal P}(w[i_{\lfloor k/2\rfloor}+1:n],\lfloor k/2\rfloor)$ using $d(n-i_{\lfloor k/2\rfloor})\leq dn$ space (which is, once more, reused space). Thus, $S(n,k)$ is upped bounded by $dn$. 

This concludes our proof, and shows the statement of the theorem for the case of increasing the universality index of a word $w$ by insertions to a value $k$ with $k\leq |w|$. 

We immediately get that in the case of increasing the universality index of a word $w$ by insertions to a value $k$ with $k> |w|$ we need $O(n)$ space to reach $n$-universality, and the rest of the construction can be done trivially in $O(k\sigma)$ space (we just need to write down the output). 

The case of decreasing the universality index of a word $w$ by deletions to a value $k<\iota(w)$ can be treated in an identical way (using the same divide-and-conquer trick). 
\end{proof}

\begin{theorem}\label{distance-binary}
Let $w$ be a word, with $|w| = n$ and $alph(w)=\{\ta,\tb\}$. Let $k\geq \iota(w)$ be a positive integer.
We can compute in $O(|w|)$ time the minimal number of insertions needed to obtain a $k$-universal~word from a given binary word $w$. 
%Particularly, if $\ell$ is the universality index of $w$, we have
%$i=0$ if $k\leq \ell$, $i=k-\ell$ if $\ell <k\leq n-2\ell$, and $i=2k-n$ if $n-2\ell <k$.
\end{theorem}
\begin{proof}
We show that, if $\ell$ is the universality index of $w$, we have that the number $i$ of insertions, that we want to compute, is $i=0$ if $k\leq \ell$, $i=k-\ell$ if $\ell <k\leq n-2\ell$, and $i=2k-n$ if $n-2\ell <k$.

Define the mapping $\overline{\cdot}:\Sigma\rightarrow\Sigma$ by $\overline{\ta}=\tb$ and $\overline{\tb}=\ta$.
The claim holds immediately for $k\leq \ell$ by \autoref{theodlt}. Assume $k>\ell$. Since $\{\ta\tb,\tb\ta\}^{\ell}\subseteq\ScatFact_{2\ell}(w)$ there exists $w_1,\dots,w_{\ell}\in\{\ta\tb,\tb\ta\}$ and $r_1,\dots,r_\ell\in[n]$ with $w[r_j,r_j+1]=w_j$ for all $j\in[\ell-1]$.
Notice that by the choice of $\ell$ we have $r_j+1<r_{j+1}$.
Set $u_1=w[1:r_1-1]$, $u_{\ell+1}=w[r_{\ell}+2:n]$, and $u_j=w[r_{j-1}+2:r_j-1]$ with $u_s=\varepsilon$ if the first index is strictly greater than the second, for $2\leq j\leq \ell$ and $s\in[\ell]$. This implies $w=u_1w_1u_2\dots u_{\ell}w_{\ell}u_{\ell+1}$ and $u=u_1\dots u_{\ell+1}$ is of length $n-2\ell$. If $n-2\ell\geq k>\ell$ set
\[
u'=u[1]\overline{u[1]} u[2]\overline{u[2]}\dots 
u_{k-\ell}\overline{u[k-\ell]} u[k-\ell+1]\dots u[n-2\ell]
\]
and $u_s'$ accordingly for $s\in[\ell+1]$. Then $w'=u_1'w_1u_2'\dots u_{\ell}'w_{\ell}u_{\ell+1}'$ is obtained from $w$ by $k-\ell$ insertions and by the definition of $u'$ we have $\{\ta\tb,\tb\ta\}^k\subseteq\ScatFact_{2k}(w')$. This implies that $w'$ is $k$-universal. If $k>n-2\ell$ set
\[
u'=u[1]\overline{u[1]} u[2]\overline{u[2]}\dots 
u_{k-\ell}\overline{u[k-\ell]} u[k-\ell+1] \overline{u[k-\ell+1]} \dots u[n-2\ell]\overline{u[n-2\ell]}
\]
and $u_s'$ accordingly for $s\in[\ell+1]$. Then $w'=u_1'w_1u_2'\dots u_{\ell}'w_{\ell}u_{\ell}'(\ta\tb)^{k-n+\ell}$ is obtained from $w$ by $n-2\ell+2(k-n+\ell)=2k-n$ insertions. By definition of $u'$ and the appended $(\ta\tb)^{k-n+\ell}$ we get $\{\ta\tb,\tb\ta\}^{\ell+n-2\ell+k-n+\ell}=\{\ta\tb,\tb\ta\}^k$ is a subset of $\ScatFact_k(w')$ and by \autoref{theodlt}, $w'$ is $k$-universal. This proves that with $i$ insertions a $k$-universal word can be obtained from $w$.

\smallskip

We prove now that $i$ is minimal. Suppose that $i$ is not minimal and let $i'<i$ be the minimal number of insertions such that the obtained word $w'$ is $k$-universal. By \autoref{theodlt} we have $\{\ta\tb,\tb\ta\}^{k}\subseteq\ScatFact_{2k}(w')$ and there exists $w_1,\dots,w_k\in\{\ta\tb,\tb\ta\}$ such that $w_1\dots w_k$ is a subsequence of $w'$. Let $j''$ be the number of $w_s$ which were inserted completely and $j'$ be the number of $w_s$ in which one letter was already in $w$ and one is inserted, i.e. $i'=j'+2j''$. This implies $k\leq\ell+j'+j''$. By the first part of the proof we have $k= i+\ell$. If $\ell<k\leq n-2\ell$ we get $i'+\ell<i+\ell= k\leq\ell+j'+j''$ and thus $i'<j'+j''$ which contradicts $i'=j'+2j''$. If $k>n-2\ell$ we get with $n+i'\geq 2\ell+2j''+2j'$ (for $w'$'s length)
\[
2(\ell+j'+j'')\leq n+i'<n+i=2k=2(\ell+j'+j'').
\]
Hence, $i$ is minimal.

\smallskip

By \autoref{decomp} the decomposition into $w_1,\dots, w_{\ell}$ can be found in run-time $\mathcal{O}(n)$. The word $u_s'$ for $s\in[\ell+1]$ can be constructed by going once one left to right while inserting after each letter $x$ the {\em opposite} letter $\overline{x}$ until $k-\ell$ insertions are reached. If $k>n-2\ell$ for each letter a new one is inserted and the remaining $(\ta\tb)^{k-n+\ell}$ letter are simply appended. 
\end{proof}

\section{Appendix: The computational model}\label{RAM}

{\bf General algorithmic framework:} Our results are of algorithmic nature. The computational model we use is the standard unit-cost RAM with logarithmic word size: for an input of size $n$, each memory word can hold $\log n$ bits. Arithmetic and bitwise operations with numbers in $[1:n]$ are, thus, assumed to take $O(1)$ time. Numbers larger than $n$, with $\ell$ bits, are represented in $O(\ell/\log n)$ memory words, and working with them takes time proportional to the number of memory words on which they are represented. In all the problems, we assume that we are given a word $w$, with $|w|=n$, over an alphabet $\Sigma=\{1,2,\ldots,\sigma\}$, with $|\Sigma|=\sigma\leq n$. That is, we assume that the processed words are sequences of integers (called letters or symbols, each fitting in $O(1)$ memory words). This is a common assumption in string algorithms: the input alphabet is said to be {\em an integer alphabet}. For a more detailed general discussion on this model see, e.g.,~\cite{crochemore}, and for its use in the particular case of algorithms related to subsequences see~\cite{mfcs2020}.

It is interesting to see whether we can extend our results for more general input alphabets, so when we drop the assumption that if the input word is $w$, then $w$ is over the alphabet $\Sigma=\{1,\ldots,
\sigma\}$ with $\sigma\leq |w|$. In the next paragraph we will follow the similar discussion made in \cite{mfcs2020}.

A computational model used in string algorithms assumes that the input is over general ordered alphabets (see \cite{dimaRuns,dimaCF,pawelLCE} and the references therein). More precisely, the input is a sequence of elements from a totally ordered set ${\mathcal U}$ (i.e., string over ${\mathcal U}$). The operations allowed in this model are those of the standard Word RAM model, with one important restriction: the elements of the input cannot be directly accessed nor stored in the memory used by the algorithms; instead, we are only allowed to {\em compare} (w.r.t. the order in ${\mathcal U}$) any two elements of the input, and the answer to such a comparison-query is retrieved in $O(1)$ time. In this model, it holds that sorting the elements of an input sequence requires at least $\Omega(n \log n)$ comparisons. An implementation of our algorithms, where we first sort the letters of the input word, map them to words over $\{1,\ldots,n\}$, and then use the same strategies as the ones described for the case of integer alphabets, would require some additional $O(n \log n)$ computational time, due to the sorting.

In fact, one cannot hope to go under $\Omega(n \log n)$ comparisons in the respective model of computation. Indeed, in this framework, it also holds that testing the equality of two sets of size $O(n)$ requires $\Omega(n \log n)$ comparisons~\cite{dobkin}. We can show the following lower bounds.

 \begin{theorem} \label{lower}
Let $w$ be a word, with $|w|=n$, $\letters(w)=\Sigma$, and universality index $\iota(w)$. Let $k$ be an integer with $n\geq k$. Computing the minimal number of insertions (respectively, deletions, or substitutions) needed to transform $w$ into a $k$-universal word requires $\Omega(n \log n)$ comparisons (so $\Omega(n \log n)$ time as well). 
\end{theorem}
\begin{proof}
Let $S=\{s_1,\ldots,s_n\}$ and $T=\{t_1,\ldots,t_m\}$ two sets, with $m\leq n$. We define the alphabet $\Sigma = S\cup T\cup \{\$,\#\}$, where the letters $\$ $ and $\# $ do not occur in $S$. We define the word $$w=\$ s_1\cdots s_n \# \$ t_1\cdots t_m \#.$$

We want to show that $S=T$ if and only if the number of insertions needed to transform $w$ into a $2$-universal word is $0$. 

The left to right implication is trivial. The right to left implication is also easy to show. If $w$ is $2$-universal, then it has two arches. These arches must be $\$ s_1\cdots s_n \# $ and $\$ t_1\cdots t_m \#$ (otherwise one would need to insert one of the separators). This means that the letters $s_1,\ldots,s_n$ are the same as $t_1,\ldots, t_m$. So $S=T$. 

Thus, to check the equality $S=T$ we can compute the minimal number of insertions needed to make $w$ $2$-universal. Thus, this requires at least $\Omega(|w|\log |w|)$ comparisons. As $|w|=n+m+4\in O(n)$, the statement follows. 

We can similarly show that $S=T$ if and only if the minimal number of substitutions needed to transform $w$ into a $2$-universal word is $0$. Similarly to the case of insertions, the lower bound is easily obtained now.

Finally, we can show that $S\neq T$ if and only if the minimal number of deletions needed to transform $w'= s_1\cdots s_n  t_1\cdots t_m$ into a $0$-universal word (w.r.t. $\letters(w')$) is exactly $1$. Indeed, each element $s_i$ occurs exactly once in $S$ and each element $t_j$ occurs exactly once in $T$. So, each letter $s_i$ (respectively, $t_j$) may occur at most twice in $w'$ (if it is contained in both $S$ and $T$). Clearly, if there is a letter that occurs exactly once, then  the minimal number of deletions needed to transform $w'= s_1\cdots s_n  t_1\cdots t_n$ into a $0$-universal word is $1$, but also this letter occurs only in $S$ or only in $T$. So $T\neq S$. If all letters occur twice, then the minimal number of deletions needed to transform $w'= s_1\cdots s_n  t_1\cdots t_n$ into a $0$-universal word is $2$, and $T= S$.

Thus, to check the equality $S=T$ we can compute the minimal number of deletions needed to make $w'$ $0$-universal. Thus, this requires at least $\Omega(|w|\log |w|)$ comparisons as well. As $|w|=n+m\in O(n)$, the statement follows. 
\end{proof}

So, having faster algorithms in this model of computation requires finding better methods than the dynamic programming approach we used.

In an intermediate model, we can assume that the input is a sequence of elements from a totally ordered set ${\mathcal U}$ (i.e., string over ${\mathcal U}$) whose elements can be stored in a constant number of memory words. The operations allowed in this model are those of the standard Word RAM model. So, in other words, the letters are from $[1:n^d]$ for some constant $d$, if the input word has length $n$. An implementation of our algorithms, where we first sort the letters of the input word, map them to words over $\{1,\ldots,n\}$, and then use the same strategies as the ones described for the case of integer alphabets, runs in exactly the same complexity as stated in the main part of the paper, as a set of $n$ numbers from $[1:n^d]$, where $d$ is a constant, can be sorted in $O(n)$ time using Radix-sort.

\end{appendices}

\end{document}

%% file: prels.tex
% !TeX spellcheck = en_GB
\section{Preliminaries}\label{prels}
%%%%%%%%%%%%%%%%%%%%%%%%%%%%
Let $\N$ be the set of natural numbers and $\N_0=\N\cup\{0\}$.
Define for $i,j\in\N_0$ with $i< j$ the interval $[i:j]$ as $\{i,i+1,\dots,j-1,j\}$.
An alphabet $\Sigma$ is a nonempty finite set of symbols called {\em letters}. 
A {\em word} is a finite sequence of letters from $\Sigma$, thus an element of 
the free monoid $\Sigma^{\ast}$. Let $\Sigma^+=\Sigma^{\ast}\backslash\{\varepsilon\}$, where $\varepsilon$ is the empty word. The {\em length} of 
a word $w\in\Sigma^{\ast}$ is denoted by $|w|$. Let $\Sigma^k$ be the set of all words from $\Sigma^*$ of length exactly $k$.
A word $u\in\Sigma^{\ast}$ is a {\em factor} of $w\in\Sigma^{\ast}$ if $w=xuy$ 
for some $x,y\in\Sigma^{\ast}$. If $x=\varepsilon$ (resp. $y=\varepsilon$), $u$ is called a 
{\em prefix} (resp. {\em suffix}) of $w$. The \nth{$i$} letter of $w\in\Sigma^{\ast}$ is denoted by 
$w[i]$ for $i\in[1:|w|]$.
Set $w[i:j] = 
w[i] w[i+1] \cdots w[j]$ for $1\leq i\leq j\leq |w|$,
$|w|_{\ta}=|\{i\in[1:|w|]|\,w[i]=\ta\}|$,
 and $\letters(w)$ $= \{\ta \in \Sigma | |w|_\ta > 0 \}$ for $w\in\Sigma^{\ast}$.
We can now introduce the notion of subsequence.

\begin{definition}
A word $v= v_1 \cdots v_\ell\in \Sigma^*$ is called a  {\em subsequence (or scattered factor)} of $w \in \Sigma^*$ if there exist $x_1, \ldots, 
x_{\ell+1}\in\Sigma^{\ast}$ such that $w = x_1 v_1 \cdots x_\ell v_\ell x_{\ell+1}$.
Let $\ScatFact(w)$ be the set of all subsequences of $w$ and
define $\ScatFact_k(w)=\ScatFact(w)\cap \Sigma^k$, the set of all subsequences of $w$ of 
length $k\in\N$. 
%$\ScatFact_k(w)=\{v \mid v\in \ScatFact(w), |v|=k\}$. 
\end{definition}

For $k\in\N_0$, $\ScatFact_k(w)$ is  
called the $k$-spectrum of $w$. Simon 
\cite{Simon72} defined the congruence $\sim_k$ in which $u,v\in\Sigma^{\ast}$ are congruent if they
have the same $k$-spectrum. As introduced in \cite{DBLP:conf/dlt/BarkerFHMN20} the notion of $k$-universality of a word over $\Sigma$ denotes its property of having~$\Sigma^k$ as $k$-spectrum.

\begin{definition}
A word $w\in\Sigma^{\ast}$ is called {\em $k$-universal} (w.r.t. $\Sigma$), for $k\in\N$, if $\ScatFact_k(w)=\Sigma^k$. We abbreviate $1$-universal by {\em universal}. The {\em universality-index} $\iota(w)$ of $w\in\Sigma^{\ast}$ is the largest $k$ such that $w$ is $k$-universal.
\end{definition}

\begin{remark}
If $\iota(w)=k$ then $w$ is $\ell$-universal for all $\ell\leq k$. Notice that $k$-universality is always w.r.t. a given alphabet $\Sigma$: the word $\ta\tb\tc\tb\ta$ is universal for $\Sigma=\{\ta,\tb,\tc\}$ but it is not universal for $\Sigma\cup\{\td\}$. In each algorithm presented in this paper, whenever we discuss about the universality index of some word (factor of the input word, or obtained from the input word via edit operations), we compute it with respect to the alphabet of the input word $w$.
\end{remark}

The notion of $\ell$-universality coincides to that of $\ell$-richness introduced in \cite{CSLKarandikarS,journals/lmcs/KarandikarS19}. We use the name {\em $\ell$-universality} rather than {\em $\ell$-richness}, as richness of words is also used with other meanings, see, e.g., \cite{DroubayJP01,LucaGZ08}. 
We recall the arch factorisation, introduced by Hebrard~\cite{TCS::Hebrard1991}.

%\begin{remark}\label{not-unique}
%It is worth noting that, unlike the case of factor universality of words and partial words \cite{martin1934,Bruijn46,ChenKMS17,GoecknerGHKKKS18}, in the case of subsequences it does not make sense to try to identify a $k$-universal word $w\in\Sigma^{\ast}$, for $k\in\N_0$, such that each word from $\Sigma^k$ occurs {\em exactly once} as subsequence of $w$. Indeed for $|\Sigma|=\sigma$, if $|w|\geq k+\sigma$ then there exists a word from $\Sigma^k$ which occurs at least twice as a subsequence of $w$. Moreover, the shortest word which is $k$-universal has length $k\sigma$ (we need $\ta^k\in\ScatFact_k(w)$ for all $\ta\in\Sigma$). As $k\sigma\geq k+\sigma$ for $k,\sigma\in \N_{\geq 2}$, all $k$-universal words have subsequences occurring more than once: there exists $i,j\in [1:\sigma+1]$ such that $w[i]=w[j]$ and $i\neq j$. Then $w[i]w[\sigma+2:\sigma+k], w[j]
%w[\sigma+2:\sigma+k]\in\ScatFact_k(w)$ and $w[i]w[\sigma+2:\sigma+k]=w[j]w[\sigma+2:\sigma+k]$.
%\end{remark}

\begin{definition}[\cite{TCS::Hebrard1991}]\label{archfact}
For $w\in\Sigma^{\ast}$ the {\em arch factorisation} of $w$ is $w=\ar_w(1)\cdots \ar_w(k)r(w)$ for some $k\in\N_0$ where $\ar_w(i)$ is universal,  
the last letter of $\ar_{w}(i)$, namely $\ar_w(i)[|\ar_w(i)|]$, does not occur in $\ar_w(i)[1:|\ar_w(i)|-1]$  for all $i\in[1:k]$, and 
$\letters(r(w))\subset\Sigma$. 
The words $\ar_w(i)$ are called {\em arches} of $w$, $r(w)$ is called the {\em rest}. % Set $m(w)=\ar_w(1)[|\ar_w(1)|] \cdots \ar_w(k)[|\ar_w(k)|]$ as the word containing the unique last letters of each arch.
\end{definition}

If the arch factorisation of $w$ contains $k\in\N_0$ arches, then $\iota(w)=k$. 
%Moreover if a factor $v$ of $w\in\Sigma^{\ast}$ is $k$-universal then $w$ is also $k$-universal: if $v$ has an arch factorisation with $k$ arches then $w$'s arch factorisation has at least $k$ arches (in which the arches of $v$ and $w$ are not necessarily related).
The following immediate theorem based on the work of Simon \cite{Simon72} completely characterises the set of $k$-subsequence universal words, based on Hebrard's arch factorisation.

\begin{theorem}\label{theodlt}
The word  $w\in\Sigma^{\ast}$ is $k$-universal if and only if there exist the words $v_i$, with $i\in [1:k]$, such that $v_1\cdots v_k=w$ and $\letters(v_i)=\Sigma$ for all $i\in [1:k]$.
\end{theorem}

The further preliminary results regard algorithms. We first introduce our framework. 
 
 \smallskip
 
{\bf General algorithmic framework:} The computational model we use is the standard unit-cost RAM with logarithmic word size: for an input of size $n$, each memory word can hold $\log n$ bits. In all the problems, we assume that we are given a word $w$, with $|w|=n$, over an alphabet $\Sigma=\{1,2,\ldots,\sigma\}$, with $|\Sigma|=\sigma\leq n$. 
This is a common assumption in string algorithms: the input alphabet is said to be {\em an integer alphabet}. For a more detailed general discussion on this model see, e.g.,~\cite{crochemore} or Appendix~\ref{RAM}. We also assume that our input words contain at least two distinct letters, otherwise all the problems we consider become trivial.

The following theorem was proven in \cite{DBLP:conf/dlt/BarkerFHMN20} and shows that the universality index and the arches can be obtained in linear time w.r.t. the word length. 

\begin{theorem}\label{decomp}
Let $w$ be a word, with $|w|=n$, $\letters(w)=\Sigma$, and $\Sigma = \{1,2,\ldots,\sigma\}$. We can compute in linear 
time $O(n)$ the arch factorisation of $w$, and, as such, $\iota(w)$.
\end{theorem}

More precisely, one computes greedily, in linear time, the following decomposition of $w$ into arches $w=u_1\cdots u_k$ as follows:
\begin{itemize}
\item $u_1$ is the shortest prefix of $w$ with $\letters(u_1)=\Sigma$, or $u_1=w$ if there is no such prefix;
\item if $u_1\cdots u_i = w[1:t]$, for some $i\in [1:k]$ and $t\in[1:n]$, we compute $u_{i+1}$ as the shortest
prefix of $w[t +1:n]$ with $\letters(u_{i+1}) = \Sigma$, or $u_{i+1} = w[t+1:n]$ if there is no such prefix.
\end{itemize}

We will now present two efficient data structures we use in our results. 
First, the {\em interval union-find} data structure~\cite{gabow,union-find}.
\begin{definition}[Interval union-find]\label{UnionFind}
Let $V=[1:n]$ and $S$ a set with $S \subseteq V$. The elements of $S =\{ s_1, \ldots , s_p \}$ are called borders and are ordered $0 = s_0 < s_1 < \ldots < s_p < s_{p+1} = n+1$ where $s_0$ and $s_{p+1}$ are generic borders. For each border $s_i$, we define $V(s_i) =[s_{i-1} +1: s_i]$ as an induced interval. Now, $P(S) \coloneqq \lbrace V(s_i)~|~ s_i \in S \rbrace$ gives an ordered partition of the set $V$. The {\em interval union-find structure} maintains the partition $P(S)$ under the operations:
\begin{itemize}
\item For $u \in V$, $\find(u)$ returns $s_i \in S \cup \lbrace n+1 \rbrace$ such that $u \in V(s_i)$. % In other words all elements in the interval $V(s_i)$ are represented by or have the representative $s_i$.
\item For $u \in S$, $\union(u)$ updates the partition $P(S)$ to $P(S \setminus \{ u\} )$. That is, if $u=s_i$, then we replace the intervals $V(s_i)$ and $V(s_{i+1})$ by the single interval $[s_{i-1}+1:s_{i+1}]$ and update the partition so that further $\find$ and $\union$ operations can be performed.
\end{itemize}

\end{definition}
When using the data structure from Definition \ref{UnionFind}, we employ a less technical language: we describe the intervals stored initially in the structure, and then the unions are made between adjacent intervals. We can enhance the data structures so that the $\find$ operation returns both borders of the interval containing the searched value, as well as some other satellite data we decide to associate to that interval. 
The following lemma was shown in~\cite{gabow,union-find}.  

\begin{lemma}\label{LemUnionFind}
One can implement the interval union-find data structure, such that, the initialisation of the structures followed by a sequence of $m\in O(n)$ union and find operations can be executed in $\mathcal{O}(n)$ time and space.
\end{lemma}

\vspace*{-3pt}
Finally, we recall the {\em Range Minimum Query} problem, and the main result on it~\cite{rmq}.
\begin{definition}[RMQ]\label{RMQ-def}
Let $A$ be an array with $n$ elements from a well-ordered set. We define {\em range minimum queries} $\RMQ_A$ for the array of $A$: 
$\RMQ_A(i,j)=\argmin \{A[t]\mid t\in [i:j]\}$, for $i,j\in [1:n]$. That is, $\RMQ_A(i,j)$ is the position of the smallest element in the subarray $A[i:j]$; if there are multiple positions containing this smallest element, $\RMQ_A(i,j)$ is the leftmost of them. (When it is clear from the context, we drop the subscript $A$).
\end{definition}

%We will use the following result from \cite{rmq}.
\begin{lemma}\label{RMQ}
Let $A$ be an array with $n$ integer elements. One can preprocess $A$ in $O(n)$ time and produce data structures allowing to answer in constant time {\em range minimum queries} $\RMQ_A(i,j)$, for any $i,j\in [1:n]$. 
\end{lemma}

%% file: exampleToolbox.tex
% !TeX spellcheck = en_GB
\newcommand\ua{\uparrow}

\subsection{Examples}
\label{examples}

These examples are based on (and supposed to be a companion in understanding) the algorithms and proofs in \autoref{toolbox}.

Most of the algorithms which we exemplify in this section use a temporary array $C$ to keep track of the letters occurring in $w$, but in slightly different ways. We will explain in each case what is the semantic of the elements in the array $C[\cdot]$.

For convenience, we assume $\Sigma=\{\ta,\tb,\tn\}$ for the examples instead of $\Sigma=\{1,2,3\}$. Let $w=\tb\ta\tn\ta\tn\ta\tb\ta\tn$ and thus $n=9$ and $\sigma=3$. 

In \autoref{init} we want to compute $\Delta(1,1),\dots,\Delta(1,9)$. Therefore, we traverse the word from left to right, i.e. from $\ell=1$ to $\ell=9$ and we maintain an array $C$ of length $3$ as well as a counter $f$. In this lemma, when reaching position $i$ of the word, $C[a]=1$ if and only if $|w[1:i]|_{a}\neq 0$, for $a\in \{\ta,\tb,\tn\}$. This results in the following computation:
\begin{center}
\scalebox{1}{\begin{tabular}{c|ccccccccc}
 & $\tb$ & $\ta$ & $\tn$ & $\ta$ & $\tn$ & $\ta$ & $\tb$ & $\ta$ & $\tn $ \\
 \hline
$\ell$           & 1 & 2 & 3 & 4 & 5 & 6 & 7 & 8 & 9\\ \hline
$C[\ta]$         & 0 & 1 & 1 & 1 & 1 & 1 & 1 & 1 & 1\\
$C[\tb]$         & 1 & 1 & 1 & 1 & 1 & 1 & 1 & 1 & 1\\
$C[\tn]$         & 0 & 0 & 1 & 1 & 1 & 1 & 1 & 1 & 1\\ \hline
$f$              & 1 & 2 & 3 & 3 & 3 & 3 & 3 & 3 & 3 \\ 
\hline
$\Delta(1,\ell)$ & 1 & 2 & 3 & 3 & 3 & 3 & 3 & 3 & 3 
\end{tabular}}
\ .
\end{center}

Notice that we have $\sigma=3$. Hence for \autoref{delta-sigma}, we only consider $i\in[3:9]$ and we want to compute $\Delta(1,3)$, $\Delta(2,4)$, $\Delta(3,5)$, $\Delta(4,6)$, $\Delta(5,7)$, $\Delta(6,8)$, and $\Delta(7,9)$. In this case, when processing position $i$, $C[a]=|w[i-\sigma+1,i]|_a,$ for $a\in \{\ta,\tb,\tn\}$.
\begin{center}
\scalebox{1}{\begin{tabular}{c|ccccccccc}
 & $\tb$ & $\ta$ & $\tn$ & $\ta$ & $\tn$ & $\ta$ & $\tb$ & $\ta$ & $\tn $ \\
 \hline
$i$                    & 1 & 2 & 3 & 4 & 5 & 6 & 7 & 8 & 9  \\ \hline
$C[\ta]$               & 0 & 1 & 1 & 2 & 1 & 2 & 1 & 2 & 1\\
$C[\tb]$               & 1 & 1 & 1 & 0 & 0 & 0 & 1 & 1 & 1 \\
$C[\tn]$               & 0 & 0 & 1 & 1 & 2 & 1 & 1 & 0 & 1\\ \hline
$f$                    & 1 & 2 & 3 & 2 & 2 & 2 & 3 & 2 & 3\\ 
$\Delta(1,i)$          & 1 & 2 & 3 & 3 & 3 & 3 & 3 & 3 & 3\\ \hline
$\Delta(i-\sigma+1,i)$ & - & - & 3 & 2 & 2 & 2 & 3 & 2 & 3
\end{tabular}}
\end{center}

In \autoref{last_occs_ins} we determine for all $j\leq\frac{n-1}{\sigma}=\frac{5}{3}$ the values $\last_{j\sigma+1}[a]$ and $d_{j\sigma+1}[a]$ for all $a\in\Sigma$. The way the array $C[\cdot]$ is used in this case is a bit different: when processing position $i$,  $C[a]$ is the position of the last occurrence of $a$ in $w[1:i],$ for $a\in \{\ta,\tb,\tn\}$.

{
\begin{center}
\scalebox{1}{\begin{tabular}{c|ccccccc}
 & $\tb$ & $\ta$ & $\tn$ & $\ta$ & $\tn$ & $\ta$ & $\tb$ \\
 \hline
$i$              & 1             & 2           & 3               & 4               & 5               & 6               & 7\\ \hline
$C[\ta]$         & $\infty$      & 2           & 2               & 4               & 4               & 6               & 6\\
$C[\tb]$         & 1             & 1           & 1               & 1               & 1               & 1               & 7 \\
$C[\tn]$         & $\infty$      & $\infty$    & 3               & 3               & 5               & 5               & 5 \\ \hline
$f$              & 1             & 2           & 3               & 3               & 3               & 3               & 3\\
$R$              & $(\tb)$       & $(\tb,\ta)$ & $(\tb,\ta,\tn)$ & $(\tb,\tn,\ta)$ & $(\tb,\ta,\tn)$ & $(\tb,\tn,\ta)$ & $(\tn,\ta,\tb)$\\
$P[\cdot]$       & $(\infty,1,\infty)$ & $(2,1,\infty)$ & $(2,1,3)$       & $(3,1,2)$       & $(2,1,3)$       & $(3,1,2)$       & $(2,3,1)$\\
\hline
$\last_{i}[\cdot]$ &             &             &                 & $[4,1,3]$       &                 &                 & $[6,7,5]$\\
$d_{i}[\cdot]$   &               &             &                 & $[1,3,2]$       &                 &                 & $[2,1,3]$
\end{tabular}}
\end{center}
}

In the table above, $P[a] $ is a pointer to the position of the list $R$ where $a\in \{\ta,\tb,\tn\}$ is stored.

Finally we have a look at the algorithms for \autoref{help1}. For the first algorithm (\autoref{alg:univ-v-l}) we get the following table. We show the state of the arrays $C$ after each iteration of the while-loop from {\tt step one}. In this case, after each iteration of the while-loop from {\tt step one}, $C[x]$ stores the number of occurrences of $x$ in $w[b:a]$, for $x\in \{\ta,\tb,\tn\}$.
With $k$ we simply count how many times the while-loop from {\tt step one} was executed:
\begin{center}
\scalebox{1}{\begin{tabular}{c|cccccccccc}
$k$                               & 1 & 2 & 3       & 4       & 5         & 6         & 7         & 8         & 9         & 10\\ \hline
$a$                               & 9 & 9 & 9       & 8       & 8         & 7         & 6         & 6         & 6         & 6 \\
$b$                               & 9 & 8 & 7       & 6       & 5         & 5         & 4         & 3         & 2         & 1\\ \hline
$C[\ta]$                          & 0 & 1 & 1       & 2       & 2         & 1         & 2         & 2         & 3         & 3\\
$C[\tb]$                          & 0 & 0 & 1       & 1       & 1         & 1         & 0         & 0         & 0         & 1\\
$C[\tn]$                          & 1 & 1 & 1      & 0       & 1         & 1         & 1         & 2         & 2         & 2\\ \hline
$f$                               & 1 & 2 & 3       & 2       & 3         & 3         & 2         & 2         & 2         & 3\\ \hline
$V_{\tb\ta\tn\ta\tn\ta\tb\ta\tn}$ &   &   & $\{7\}$ & $\{7\}$ & $\{7,5\}$ & $\{7,5\}$ & $\{7,5\}$ & $\{7,5\}$ & $\{7,5\}$ & $\{7,5,1\}$\\
$\univ[9]$                        &   &   & 7       & 7       & 7         & 7         & 7         & 7         & 7         & 7\\
$\univ[8]$                        &   &   &         &         & 5         & 5         & 5         & 5         & 5         & 5\\
$\univ[7]$                        &   &   &         &         &           & 5         & 5         & 5         & 5         & 5\\
$\univ[6]$                        &   &   &         &         &           &           &           &           &           & 1 \\ \hline
$L_7$                             &   &   & $\{9\}$ & $\{9\}$ & $\{9\}$   & $\{9\}$   & $\{9\}$   & $\{9\}$   & $\{9\}$   & $\{9\}$\\
$L_5$                             &   &   &         &         & $\{8\}$   & $\{8,7\}$ & $\{8,7\}$ & $\{8,7\}$ & $\{8,7\}$ & $\{8,7\}$\\
$L_1$                             &   &   &         &         &           &           &           &           &           & $\{6\}$
\end{tabular}}\ .
\end{center}

The while-loop in {\tt step three} is now used and we will obtain $\univ[5]=\univ[4]=\univ[3]=1$, and $5,4,3$ are all added in $L_1$. So $L_1=\{3,4,5,6\}$. The rest of the values in the $\univ$ array are left as initialized, namely $0$, and $L_0=\{1,2\}$.

For the second algorithm (\autoref{alg:freq-t}), when reaching position $i$ of the word, $C[a]=|w[1:i]|_{a}$, for $a\in \{\ta,\tb,\tn\}$. Thus, we compute:
\begin{center}
\scalebox{1}{
\begin{tabular}{c|ccccccccc}
 & $\tb$ & $\ta$ & $\tn$ & $\ta$ & $\tn$ & $\ta$ & $\tb$ & $\ta$ & $\tn $ \\
 \hline
$i$            & 1 & 2 & 3 & 4 & 5 & 6 & 7 & 8 & 9\\ \hline
$C[\ta]$       & 0 & 1 & 1 & 2 & 2 & 3 & 3 & 4 & 4\\
$C[\tb]$       & 1 & 1 & 1 & 1 & 1 & 1 & 2 & 2 & 2\\
$C[\tn]$       & 0 & 0 & 1 & 1 & 2 & 2 & 2 & 2 & 3\\ \hline
$\freq[\cdot]$ & 1 & 1 & 1 & 2 & 2 & 3 & 2 & 4 & 3
\end{tabular}}\ .
\end{center}
For the computation of $T$ set $x=\tb$ and $m=2$ (as determined by the fact that $C[\tb]$ is the minimum in $C$). This implies $T[0]=2$. Thus, we get for the computation of $T$:
\begin{center}
\scalebox{1}{
\begin{tabular}{c|cccccccc}
$i$            & 1     & 2     & 3     & 4     & 5     & 6     & 7     & 8 \\ \hline
$C[\ta]$       & 4     & 3     & 3     & 2     & 2     & 1     & 1     & 0 \\
$C[\tb]$       & 1     & 1     & 1     & 1     & 1     & 1     & 0     & 0 \\
$C[\tn]$       & 3     & 3     & 2     & 2     & 1     & 1     & 1     & 1 \\ \hline
$m$            & 1     & 1     & 1     & 1     & 1     & 1     & 0     & 0 \\
$x$            & $\tb$ & $\tb$ & $\tb$ & $\tb$ & $\tb$ & $\tb$ & $\tb$ & $\tb$\\
$T$            & 1     & 1     & 1     & 1     & 1     & 1     & 0     & 0
\end{tabular}}\ .
\end{center}
Note that $T[9]$ is also set to $0$. Now, in each step, when considering the letter $a$, we decrement $C[a]$ by $1$, and then compute again the minimum of $C$.

For the last algorithm (\autoref{alg:last}) we get, with $V_{\tb\ta\tn\ta\tn\ta\tb\ta\tn}=\{7,5,1\}$, $L_1=\{3,4,5,6\}$ and $w[3]=\tn$, $ w[4]=\ta$, $w[5]=\tn$, $w[6]=\ta$, $L_5=\{7,8\}$ and $w[7]=\tb$ and $w[8]=\ta$, $L_7=\{9\}$ and $w[9]=\tn$:
\begin{center}
\scalebox{1}{\begin{tabular}{c|cccccccccc}
$i$       &    & 0  & 1  & 2 & 3 & 4 & 5 & 6 & 7 & 8\\ \hline
$\last_i [\cdot]$ & & [10,10,10] & [10,-,-] & [-,-,10] & [2,-,-] & [4,1,3] & [4,-,-] & [-,1,5] & [6,-,-] & [-,-,5] \\ \hline
$L[\ta]$  & 10 & 10 & 2  & 2 & 4 & 4 & 6 & 6 & 8 & 8\\
$L[\tb]$  & 10 & 1  & 1  & 1 & 1 & 1 & 1 & 7 & 7 & 7\\
$L[\tn]$  & 10 & 10 & 10 & 3 & 3 & 5 & 5 & 5 & 5 & 9
\end{tabular}}.
\end{center}
 In the table above, $\last_i [\cdot]=(\last_i[\ta], \last_i[\tb], \last_i[\tn])$. This representation is chosen for the ease of understanding. However, note that each $\last_i[\cdot]$ is implemented as a list with exactly one element if $i+1\notin V_w$ (i.e., $\last_i[w[i+1]]$) and exactly $\min\{|\Sigma|, |L_{i+1}|+1\}$ elements (i.e., $\last_i[w[i+1]]$ and $\last_i[w[j]]$ for $j\in L_{i+1}$) if $i+1\in V_w$. 

%% file: substitutions.tex
% !TeX spellcheck = en_GB
\subsection{Substitutions}\label{sub}

\begin{theorem}\label{edit-subs}
Let $w$ be a word, with $|w| = n$, $alph(w)=\Sigma $, and $\Sigma= \lbrace 1, 2, \ldots, \sigma \rbrace$. Let $k$ be an integer $0\leq k \leq \lfloor \frac{n}{\sigma} \rfloor $. We can compute the minimal number of substitutions needed to apply to $w$ in order to obtain a $k$-universal word (w.r.t. $\Sigma$) in $O(nk)$ time.
\end{theorem}
\begin{proof}
Recall that $\iota(w)$ is the initial universality index of $w$. We will distinguish between the cases $\iota(w) < k$ and $\iota(w) > k$. While the former allows for an argumentation similar to \autoref{ins-distance} for insertions, the latter will fall back to the \autoref{deletions} for deletions. 

{\bf Case 1.} Let us assume first that $\iota(w) < k$.

\smallskip

{\emph{$\S$ General approach.}} At a high level, the algorithm and data structures we use here are similar to those used in the case of changing the universality of a word by insertions, described in \autoref{ins-distance} (the finer details are, however, different). As in the respective algorithm, we will compute $M[\ell][t]$ the minimal number of substitutions one needs to apply to $w[1:\ell]$ in order to make it $t$-universal, for all $\ell\in [1:n]$ and all $t\in [1:k]$. Clearly, to edit $w[1:\ell]$ into a $t$-universal word using substitutions, we first create a $(t-1)$-universal word from a prefix $w[1:\ell']$ of $w[1:\ell]$, and then a $1$-universal word from $w[\ell'+1:\ell]$. As in the case of insertions, the number of substitutions used in this process has to be minimal among all the numbers we obtain when choosing $\ell'$ in all possible ways. 
The main differences are that, in the case of substitutions, we need to have that $|w[\ell'+1:\ell]|\geq \sigma$, or we would not be able to obtain a $1$-universal word from $w[\ell'+1:\ell]$, and $|w[1:\ell']|\geq (t-1)\sigma$.
For the first case of this proof see also \autoref{alg:substitutions-case-one}.

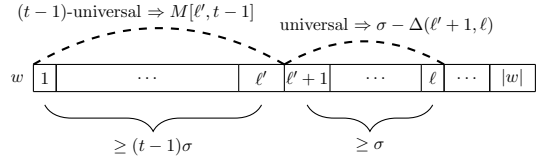
\begin{figure}[h!]
    \centering
    \begin{minipage}[c]{0.5\linewidth}
	%\centering
    	\begin{tikzpicture}[scale=0.6,every node/.style={scale=0.6}]
          % Dialectics
          
            \node[circle] (t0) at (0,0) {$w$};
            \node[wens, minimum height=6mm, minimum width=5mm, align=center, right= 0mm of t0] (t1) {$1$};
            \node[wens, minimum height=6mm, minimum width=40mm, align=center, right= 0mm of t1] (t2) {$\ldots$};
            \node[ens, minimum height=6mm, minimum width=10mm, align=center, right= 0mm of t2] (t4) {$\ell'$};
            \node[ns, minimum height=6mm, minimum width=10mm, align=center, right= 0mm of t4] (t5) {$\ell'+1$};
            \node[wns, minimum height=6mm, minimum width=20mm, align=center, right= 0mm of t5] (t7) {$\ldots$};
            \node[wens, minimum height=6mm, minimum width=5mm, align=center, right= 0mm of t7] (t8) {$\ell$};
            \node[wens, minimum height=6mm, minimum width=10mm, align=center, right= 0mm of t8] (t9) {$\ldots$};
            \node[wens, minimum height=6mm, minimum width=10mm, align=center, right= 0mm of t9] (t10) {$|w|$};

            \draw[bend left, thick, dashed] (t1.north west) edge node[above, xshift=-5mm]{$(t-1)$-universal $\Rightarrow M[\ell',t-1]$} (t4.north east);
            
            \draw[bend left, thick, dashed] (t5.north west) edge node[above, xshift=5mm]{universal $\Rightarrow \sigma - \Delta(\ell'+1,\ell)$} (t8.north east);
            
            \draw[decoration={brace,mirror,raise=5pt, amplitude=10pt},decorate] (t1.south) --  node[below=15pt]{$ \geq (t-1)\sigma$} (t4.south);
            
            \draw[decoration={brace,mirror,raise=5pt, amplitude=10pt},decorate] (t5.south) --  node[below=15pt]{$ \geq \sigma$} (t8.south);
            
        \end{tikzpicture}
    \end{minipage}
    \caption{Illustration of the formula developed for the computation of $M[\ell][t]$.}
    \label{fig:substidea}
\end{figure}

{\emph{$\S$ Algorithm - initial idea.}}
As described informally above, we compute the $n\times k$ matrix $M[\cdot][\cdot]$, where $M[\ell][t]$ denotes the minimal number of substitution needed to transform $w[1:\ell]$ into a $t$-universal word, for all $\ell \in [1:n]$ and $t\in [1:k]$. We will achieve this using dynamic programming. By the remarks we made above, it is not hard to see that:
$$M[\ell][t] = \min \lbrace M[\ell'][t-1] + (\sigma - \Delta(\ell'+1,\ell)) \mid \sigma (t-1) < \ell'+1 \leq \ell-\sigma+1 \rbrace . $$

In order to compute $M$ efficiently, we will run the algorithms from \Cref{init,last_occs_ins}. These will provide us with the data structures $\last_{j\sigma+1}[\cdot]$, and $d_{j\sigma+1}[\cdot]$, for all $j$, as well as all $M[\cdot][1]$ in $O(n)$ time. More precisely, $M[\ell][1]=\sigma - \Delta(1,\ell)$, if $\ell\geq \sigma$, and $M[\ell][1]=\infty$, if $\ell<\sigma$. As in the case of insertions, the direct computation of $M[\ell][t]$ would be too costly. Therefore, we will need further observations.

\smallskip

{\emph{$\S$ Observations.}} Assume that the letter on position $\ell'+1<\ell-\sigma+1$ in $w$, namely $w[\ell'+1]$, occurs at least twice in $w[\ell'+1:\ell]$. Then $(\sigma - \Delta(\ell'+1,\ell)) = (\sigma - \Delta(\ell'+2,\ell))$. Thus, it clearly follows that  $M[\ell'][t-1]+(\sigma - \Delta(\ell'+1,\ell)) \geq M[\ell'+1][t-1]+(\sigma - \Delta(\ell'+2,\ell))$. So, once more, we only need to consider in our recurrence that $\ell'+1$ is the rightmost occurrence of some letter inside the factor $w[1:\ell]$. The observation made above does not include the case when $\ell'+1=\ell-\sigma+1$, so we will just add this position in the set of  the relevant positions for our recurrence too. 

As we did in the proof of \autoref{ins-distance}, we can now rewrite our recurrence in a way that will enable us to apply \autoref{last_occs_ins}. For $\mathfrak{S}_\ell = (S_\ell\cap[(t-1)\sigma:\ell-\sigma])\cup\{\ell-\sigma+1\}:$
$$ M[\ell][t]= \min \{M[\ell'][t-1]+(\sigma-\Delta(\ell'+1,\ell)) \mid \ell'+1\in \mathfrak{S}_\ell\}.$$

\begin{figure}[h!]
    \centering
    \begin{minipage}[c]{0.5\linewidth}
	%\centering
    	\begin{tikzpicture}[scale=0.6,every node/.style={scale=0.6}]
          % Dialectics
          
            \node[circle] (t0) at (0,0) {$w$};
            \node[wens, minimum height=6mm, minimum width=25mm, align=center, right= 0mm of t0] (t1) {$\ldots$};
            \node[wens, fill=black!25, minimum height=6mm, minimum width=5mm, align=center, right= 0mm of t1] (t2) {$ $};
            \node[wens, minimum height=6mm, minimum width=10mm, align=center, right= 0mm of t2] (t3) {$\ldots$};
            \node[wens, fill=black!25, minimum height=6mm, minimum width=5mm, align=center, right= 0mm of t3] (t4) {$ $};
            \node[wens, minimum height=6mm, minimum width=10mm, align=center, right= 0mm of t4] (t5) {$\ldots$};
            \node[wens, fill=black!25, minimum height=6mm, minimum width=5mm, align=center, right= 0mm of t5] (t7) {$ $};
            \node[wens, minimum height=6mm, minimum width=15mm, align=center, right= 0mm of t7] (t8) {$\ldots$};
            \node[wens, fill=black!25, minimum height=6mm, minimum width=5mm, align=center, right= 0mm of t8] (t9) {$ $};
            \node[wens, minimum height=6mm, minimum width=15mm, align=center, right= 0mm of t9] (t10) {$\ldots$};
            \node[wens, minimum height=6mm, minimum width=5mm, align=center, right= 0mm of t10] (t11) {$\ell$};
            \node[wens, minimum height=6mm, minimum width=10mm, align=center, right= 0mm of t11] (t12) {$\ldots$};
            \node[minimum height=6mm, minimum width=35mm, align=center, below= 5mm of t4] (t13) {$S_\ell \cap [(t-1)\sigma:\ell-\sigma]$};
            \node[e,minimum height=6mm, minimum width=15mm, align=center, below= 0mm of t1] (t1sub) {$(t-1)\sigma$};
            \node[w,minimum height=6mm, minimum width=15mm, align=center, below= 0mm of t9, xshift=5mm] (t9sub) {$\ell-\sigma+1$};

            \draw[->] (t13) -- (t4.south);
            \draw[->] (t13) -- (t2.south);
            \draw[->] (t13) -- (t7.south);
            
            \draw[bend left, thick, dashed] (t1.north west) edge node[above, xshift=-5mm]{$\rightarrow (t-1)$-universal} (t4.north west);
            
            \draw[bend left, thick, dashed] (t4.north west) edge node[above, xshift=5mm]{$\rightarrow$ universal} (t11.north east);
            
            %\draw[decoration={brace,mirror,raise=5pt, amplitude=10pt},decorate] (t1.south) --  node[below=15pt]{$ \geq (t-1)\sigma$} (t4.south);
            
            %\draw[decoration={brace,mirror,raise=5pt, amplitude=10pt},decorate] (t5.south) --  node[below=15pt]{$ \geq \sigma$} (t8.south);
            
        \end{tikzpicture}
    \end{minipage}
    \caption{Illustration of $\mathfrak{S}_\ell$ and how it is used to compute $M[\ell][t]$. Here we only select from $S_\ell=\{\last_\ell [a]\mid a\in \letters(w[1:\ell])\}$ the elements which are in $[(t-1)\sigma:l-\sigma]$. To obtain $\mathfrak{S}_\ell$ we also have to consider the position $\ell-\sigma+1$. The elements of $\mathfrak{S}_\ell$ are depicted with grey in the figure. }
    \label{fig:subsl}
\end{figure}
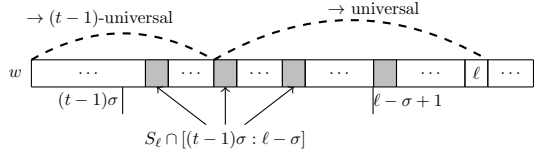

{\emph{$\S$ Algorithm - the efficient variant.}} Using these observations in combination with \autoref{speedUp}, we can compute the elements of the matrix $M$ efficiently. The algorithm is similar to the one from Case 1 of \autoref{ins-distance}. Firstly, by \autoref{delta-sigma}, we can compute the values $\Delta(i-\sigma+1, i)$, for all $i$, in $O(n)$ time.

So, let us consider a value $t\geq 2$. Assume that we have computed the values $M[\ell][t-1]$, for all $\ell\in [1:n]$. We now want to compute the values $M[\ell][t]$, for all $\ell\in [1:n]$.  
The main idea in doing this efficiently is to split the computation of the elements on column $M[\cdot][t]$ of the matrix $M$ in phases. In phase $j$ we compute the values $M[j\sigma+1][t],$ $M[j\sigma+2][t], \ldots, $ $M[(j+1)\sigma][t]$, for $j\leq (n-1)/\sigma$. 

Now we consider some $j$, with $0\leq j\leq (n-1)/\sigma$. We want to apply \autoref{speedUp}, so we need to define the list $A$ of size $\sigma$. This is done as follows. 

We will keep an auxiliary array $\pos[\cdot]$ with $\sigma$ elements. Moreover, the element on each position $i$ of $A$, namely $A[i]$, will be accompanied by two satellite data: a position of $w$ and the letter on that position. 
Now, for $a$ such that $\last_{j\sigma+1}[a] \in \mathfrak{S}_{j\sigma+1} \setminus \{(j-1)\sigma+2\}$, we have that $d_{j\sigma+1}[a]=\sigma - i$ for some $i\in[0:\sigma-2]$ (as $w[\last_{j\sigma+1}[a]:j\sigma+1]$ contains at least $a$ and $w[(j-1)\sigma+2]$). We set $A[i+1]=M[\last_{j\sigma+1}[a]-1][t-1]+i$ and $\pos[a]=i+1$; the satellite data of $A$ is the pair $(\last_{j\sigma+1}[a],a)$. 
We also set $A[\sigma]=M[(j-1)\sigma+1][t-1]+(\sigma-\Delta((j-1)\sigma+2, j\sigma+1))$ and $\pos[w[(j-1)\sigma+2]]=\sigma$; the satellite data of $A$ is the pair $((j-1)\sigma+2,w[(j-1)\sigma+2])$.
If, for some letter $a$, $\last_{j\sigma+1}[a]=n+1$ (i.e., $a$ does not occur in $w[1:j\sigma+1]$), we simply set $\pos[a]=0$. 

The elements of $A$ which are not defined above will store $\infty$ (i.e., a large enough value, at least $(k+1)\sigma+1$); for simplicity, if we apply any arithmetic operation to $\infty$, we get $\infty$. %Note that, in this case, the elements of $A$ which are not defined do not occur all on the last positions of $A$.
We also set $m$ to $\sigma$.  Now we are in the position of defining and applying a sequence of operations $o_1,\ldots,o_\sigma$ from \autoref{speedUp}.

\smallskip

{\emph{$\S$ Algorithm - application of \autoref{speedUp}}}.  In $o_i$, we extract the minimum $q$ of $A$. Then set $M[j\sigma+i][t]= q$. We decrement by $1$ all elements of $A$ on the positions $\pos[a]+1,\pos[a]+2, \ldots, m$, where $a=w[(j-1)\sigma+i+2]$. Then, we append to $A$ the element $M[(j-1)\sigma+i+1][t-1]+(\sigma-\Delta((j-1)\sigma+i+2,j\sigma+i+1))$, with the satellite data $(j\sigma+i+2,a)$ (and implicitly increment $m$ by $1$), and set $\pos[a]=m$. 

\smallskip

{\emph{$\S$ Algorithm - the result of applying \autoref{speedUp}}}.  One can show exactly as in the case of \autoref{ins-distance} that after executing the $\sigma$ operations $o_1,\ldots, o_\sigma$, we have computed the values $M[j\sigma+1][t],$ $M[j\sigma+2][t], \ldots, $ $M[(j+1)\sigma][t]$ correctly. We can move on to phase $j+1$ and repeat this process. 

\smallskip

{\emph{$\S$ The result.}} The minimal number of substitutions needed to make $w$ $k$-universal is correctly computed as $M[n][k]$. 

\smallskip 

{\emph{$\S$ Complexity.}} By \autoref{speedUp}, computing $M[j\sigma+1][t],$ $M[j\sigma+2][t], \ldots, $ $M[(j+1)\sigma][t]$ takes $O(\sigma)$ for each $j$. Overall, computing the entire column $M[\cdot][t]$ takes $O(n)$ time. We do this for all $t\leq k$, and we obtain $O(nk)$ time in total to compute all elements of the matrix $M$.

\smallskip

\SetKw{KwOr}{or}
\SetKw{KwAnd}{and}
\begin{algorithm}[H]
	\SetKwInOut{Input}{Input}
	\SetKwInOut{Output}{Output}
	\Input{word $w$, alphabet $\Sigma$, int $k$}
	\Output{minimal number of substitutions}
	\BlankLine
	
	\tcp{initialisation}
	int $n \leftarrow \lvert w \rvert$; int $\sigma \leftarrow \lvert \Sigma \rvert$\;
	int $M[n][k]=\infty$\;
	\tcp{initialise first column of $M$}
	\For{$l = \sigma$ \KwTo $n$}{
		$M[l][1] \leftarrow \sigma - \Delta(1,l)$\;
	}
	\tcp{efficient variant}
	\For{$t = 2$ \KwTo $k$}{
		
		\tcp{$\leq (n-1)/\sigma$ phases}
		\For{$j=0$ \KwTo $(n-1)/\sigma$}{
			int $A[\sigma][3]$ (list of triples including satellite data); int $\pos[\sigma]$\;
			\For{$a = 1$ \KwTo $\sigma$}{
				\tcp{$S_\ell=\{\last_\ell [a]\mid a\in \letters(w[1:\ell])\}$ }
				\tcp{$\mathfrak{S}_\ell = (S_\ell\cap[(t-1)\sigma:\ell-\sigma])\cup\{\ell-\sigma+1\}$}
				\If{$a \in \mathfrak{S}_{j\sigma+1} \setminus \{(j-1)\sigma + 2\}$}{
					int $i \leftarrow \sigma - d_{j\sigma + 1}[a]$\;
					$A[i+1][1] \leftarrow M[\last_{j\sigma+1}[a]-1][t-1]+i$\;
					$\pos[a] \leftarrow i + 1$\;
					\tcp{satellite data for $A[i+1]$}
					$A[i+1][2] \leftarrow \last_{j\sigma+1}[a]$\;
					$A[i+1][3] \leftarrow a$\;
					
					$A[\sigma][1] \leftarrow M[(j-1)\sigma + 1][t-1] + (\sigma - \Delta((j-1) \sigma + 2, j\sigma + 1))$\;
					$\pos[w[(j-1)\sigma +2]] \leftarrow \sigma$\;
					\tcp{satellite data for $A[\sigma]$}
					$A[\sigma][2] \leftarrow (j-1)\sigma + 2$\;
					$A[\sigma][3] \leftarrow w[(j-1)\sigma + 2]$\;
				}
				
				\If{ $\last_{j\sigma+1}[a] = n+1$}{
					$\pos[a] = 0$\;
				}
			}
			
			\tcp{apply sequence of operations as in \autoref{speedUp}}
			\For{$i = 1$ \KwTo $\sigma$}{
				$q \leftarrow$ minimum of $A$\;
				$M[j\sigma+i][t] \leftarrow q$\;
				$a = w[(j-1)\sigma + i + 2]$\;
				decrement positions $\pos[a]+1, \pos[a]+2, \ldots, m$ by $1$\;
				append $M[(j-1)\sigma + i+1][t-1]+(\sigma - \Delta((j-1)\sigma + i+2, j\sigma + i + 1))$ to $A$ (that is, set $A[m+1][1]$ to that value)\;
				\tcp{and add satellite data}
				$A[m+1][2] \leftarrow j\sigma + i + 2$\;
				$A[m+1][3] \leftarrow a$\;
				$m \leftarrow m + 1$\;
				$\pos[a] \leftarrow m$\;
			}
		}
	}
	
	\Return $M[n][k]$\;
	
	\caption{The efficient algorithm from Case 1 of \autoref{edit-subs} (on substitutions).}
	\label{alg:substitutions-case-one}
\end{algorithm}
\newpage

{\bf Case 2.} Now let us assume that $\iota(w) > k$.

\smallskip

{\emph{$\S$ General approach.}} 
We will show that the minimal number of substitutions needed to obtain a $k$-universal word from $w$ equals the minimal number of deletions needed to obtain a $k$-universal word from $w$, and then use the algorithm from \autoref{deletions} to compute it.

{\emph{$\S $ The proof.}} First of all, the comments made in the opening of \autoref{editDist} explain why the minimal number of substitutions needed to obtain a $k$-universal word from $w$ is lower bounded by the minimal number of deletions needed to obtain a $k$-universal word from $w$.  
Indeed, in a transformation of $w$, of universality index $\iota(w)$, into a word $w'$, of universality index $k$, using $s$ substitutions, we can replace all these substitutions by deletions and get a word $w''$ that has universality index lower or equal to $k$. Thus reaching a $k$-universal word from $w$ requires $s$ deletions or less. 

To see that, in fact, the minimal number of substitutions needed to obtain a $k$-universal word from $w$ equals the minimal number of deletions needed to obtain a $k$-universal word from $w$, we proceed as follows.

Let $ar_w(1) \cdots ar_w(\iota(w))r_w$ be the arch factorisation of $w$. Let $w'$ be a $k$-universal word that is obtained from $w$ using a minimal number of deletions, and let $w'=\ar_{w'}(1) \cdots \ar_{w'}(k)r_{w'}$ be its arch factorisation. Clearly, $w'$ is a subsequence of $w$, so we can identify the list of deleted positions of $w$, as well as the position $i_j$ of $w$ which corresponds to the last symbol of $\ar_{w'}(j)$ for $j\in [1:k]$; that is, $\ar_{w'}(1)\cdots \ar_{w'}(j) $ is obtained using the respective deletions from $w[1:i_j]$. Let $i_0=0$. 
Now, we can associate the deleted positions to the arches of $w'$ in the following way: if $i$ is a position of $w$ such that $w[i]$ was deleted and $i\in [i_{j-1}+1:i_j]$, then $i$ is associated to $\ar_{w'}(j)$. 

Now, in the case of substitutions, instead of deleting letters of $w$, we will replace any deleted letter $w[i]$ of $w$ by a letter different from the last letter of $\ar_{w'}(j)$, namely $\ar_{w'}(j)[|\ar_{w'}(j)|]$, where $\ar_{w'}(j)$ is the arch of $w'$ associated to position $i$. Let $w''$ be the word obtained in this way. By \autoref{archfact} $w''$ has an arch factorisation with exactly $k$ archs (the \nth{$j$} arch of this factorisation ends with the letters $\ar_{w'}(j)[|\ar_{w'}(j)|]$ from the corresponding arch of $w'$). 

This shows that the minimal number of substitutions needed to obtain a $k$-universal word from $w$ is lower or equal to the minimal number of deletions needed to obtain a $k$-universal word from $w$, and, as we have also shown the opposing inequality, these numbers must be equal. 

\smallskip

With this, the analysis of both cases is finished, and it follows that the statement of the theorem holds.
\end{proof}
%The case $k>\iota(w)$ is treated similarly to the case of changing the universality of a word by insertions, described in \autoref{ins-distance}. We define a matrix $M$, with $M[\ell][t]$ being the minimal number of substitutions one needs to apply to $w[1:\ell]$ in order to make it $t$-universal, for all $\ell\in [1:n]$ and all $t\in [1:k]$.
%Then, we derive the following recurrence: $M[\ell][t]=\min \{M[\ell'][t-1]+(\sigma-\Delta(\ell'+1,\ell)) \mid \ell'+1\in \mathfrak{S}_\ell\}$, 
%where~$\mathfrak{S}_\ell = (S_\ell\cap[(t-1)\sigma:\ell-\sigma])\cup\{\ell-\sigma+1\}$ ($S_\ell$ is defined in \autoref{toolbox}). The fact that at most $\sigma$ elements of $M[\cdot][t-1]$ are used to compute each of the elements $M[\ell][t]$ allows us to apply \autoref{speedUp} in almost the same way as we did in the algorithm of \autoref{ins-distance}, and compute all the elements of $M$ in $O(nk)$ time. The number we want to compute is $M[n][k]$. 

%For the case $k<\iota(w)$, we show that when decreasing the universality index of a word, it makes no difference whether we use substitutions or deletions. So, it is enough to use the algorithm of \autoref{deletions} 
%and compute the minimum number of deletions needed transform $w$ into a $k$-universal word, 
%and return the computed result as the answer~to~our~current~problem.

\begin{remark}
While substitutions and deletions can be used similarly to decrease the universality index of a word, we always need at least as many substitutions as insertions to increase it. To see that this inequality can also be strict, note that one insertion is enough to make $\ta\ta\tb\tb$ $2$-universal, but we need two substitutions to achieve the same result.
\end{remark}

%% file: scatteredUniversalityArxiv.bbl
\begin{thebibliography}{10}

\bibitem{DBLP:journals/siamcomp/BackursI18}
Arturs Backurs and Piotr Indyk.
\newblock Edit distance cannot be computed in strongly subquadratic time
  (unless {SETH} is false).
\newblock {\em {SIAM} J. Comput.}, 47(3):1087--1097, 2018.

\bibitem{DBLP:journals/tcs/Baeza-Yates91}
Ricardo~A. Baeza{-}Yates.
\newblock Searching subsequences.
\newblock {\em Theor. Comput. Sci.}, 78(2):363--376, 1991.

\bibitem{DBLP:conf/dlt/BarkerFHMN20}
Laura Barker, Pamela Fleischmann, Katharina Harwardt, Florin Manea, and Dirk
  Nowotka.
\newblock Scattered factor-universality of words.
\newblock In Natasa Jonoska and Dmytro Savchuk, editors, {\em Proc. DLT 2020},
  volume 12086 of {\em Lecture Notes in Computer Science}, pages 14--28, 2020.

\bibitem{rmq}
Michael~A. Bender and Martin Farach{-}Colton.
\newblock The {LCA} problem revisited.
\newblock In {\em Proc. {LATIN} 2000}, volume 1776 of {\em Lecture Notes in
  Computer Science}, pages 88--94, 2000.

\bibitem{DBLP:conf/fsttcs/BringmannC18}
Karl Bringmann and Bhaskar~Ray Chaudhury.
\newblock Sketching, streaming, and fine-grained complexity of (weighted)
  {LCS}.
\newblock In {\em Proc. {FSTTCS} 2018}, volume 122 of {\em LIPIcs}, pages
  40:1--40:16, 2018.

\bibitem{DBLP:conf/focs/BringmannGSW16}
Karl Bringmann, Fabrizio Grandoni, Barna Saha, and Virginia~Vassilevska
  Williams.
\newblock Truly sub-cubic algorithms for language edit distance and
  {RNA}-folding via fast bounded-difference min-plus product.
\newblock In {\em Proc. {FOCS} 2016}, pages 375--384, 2016.

\bibitem{BringmannK18}
Karl Bringmann and Marvin K{\"{u}}nnemann.
\newblock Multivariate fine-grained complexity of longest common subsequence.
\newblock In {\em Proc. {SODA} 2018}, pages 1216--1235, 2018.

\bibitem{ChenKMS17}
Herman Z.~Q. Chen, Sergey Kitaev, Torsten M{\"u}tze, and Brian~Y. Sun.
\newblock On universal partial words.
\newblock {\em Electronic Notes in Discrete Mathematics}, 61:231--237, 2017.

\bibitem{DBLP:conf/dlt/CheonH20}
Hyunjoon Cheon and Yo{-}Sub Han.
\newblock Computing the shortest string and the edit-distance for parsing
  expression languages.
\newblock In {\em Proc. DLT 2020}, volume 12086 of {\em Lecture Notes in
  Computer Science}, pages 43--54, 2020.

\bibitem{DBLP:conf/dlt/CheonHKS19}
Hyunjoon Cheon, Yo{-}Sub Han, Sang{-}Ki Ko, and Kai Salomaa.
\newblock The relative edit-distance between two input-driven languages.
\newblock In {\em Proc. DLT 2019}, volume 11647 of {\em Lecture Notes in
  Computer Science}, pages 127--139, 2019.

\bibitem{crochemore}
Maxime Crochemore, Christophe Hancart, and Thierry Lecroq.
\newblock {\em Algorithms on strings}.
\newblock Cambridge University Press, 2007.

\bibitem{DBLP:journals/jda/CrochemoreMT03}
Maxime Crochemore, Borivoj Melichar, and Zdenek Tron{\'{\i}}cek.
\newblock Directed acyclic subsequence graph - overview.
\newblock {\em J. Discrete Algorithms}, 1(3-4):255--280, 2003.

\bibitem{dlt2019}
Joel~D. Day, Pamela Fleischmann, Florin Manea, and Dirk Nowotka.
\newblock k-spectra of weakly-c-balanced words.
\newblock In {\em Proc. {DLT} 2019}, volume 11647 of {\em Lecture Notes in
  Computer Science}, pages 265--277, 2019.

\bibitem{Bruijn46}
Nicolaas~G. de~Bruijn.
\newblock A combinatorial problem.
\newblock {\em Koninklijke Nederlandse Akademie v. Wetenschappen}, 49:758--764,
  1946.

\bibitem{LucaGZ08}
Aldo de~Luca, Amy Glen, and Luca~Q. Zamboni.
\newblock Rich, sturmian, and trapezoidal words.
\newblock {\em Theor. Comput. Sci.}, 407(1-3):569--573, 2008.
\newblock \href {https://doi.org/10.1016/j.tcs.2008.06.009}
  {\path{doi:10.1016/j.tcs.2008.06.009}}.

\bibitem{dobkin}
David~P. Dobkin and Richard~J. Lipton.
\newblock On the complexity of computations under varying sets of primitives.
\newblock {\em J. Comput. Syst. Sci.}, 18(1):86--91, 1979.
\newblock \href {https://doi.org/10.1016/0022-0000(79)90054-0}
  {\path{doi:10.1016/0022-0000(79)90054-0}}.

\bibitem{DroubayJP01}
Xavier Droubay, Jacques Justin, and Giuseppe Pirillo.
\newblock Episturmian words and some constructions of de {L}uca and {R}auzy.
\newblock {\em Theor. Comput. Sci.}, 255(1-2):539--553, 2001.
\newblock \href {https://doi.org/10.1016/S0304-3975(99)00320-5}
  {\path{doi:10.1016/S0304-3975(99)00320-5}}.

\bibitem{KufMFCS}
Lukas Fleischer and Manfred Kufleitner.
\newblock Testing {S}imon's congruence.
\newblock In {\em Proc. {MFCS} 2018}, volume 117 of {\em LIPIcs}, pages
  62:1--62:13, 2018.

\bibitem{FreydenbergerGK15}
Dominik~D. Freydenberger, Pawel Gawrychowski, Juhani Karhum{\"{a}}ki, Florin
  Manea, and Wojciech Rytter.
\newblock Testing k-binomial equivalence.
\newblock In {\em {\em Multidisciplinary Creativity}, a collection of papers
  dedicated to G. P\u aun 65th birthday}, pages 239--248, 2015.
\newblock available in CoRR abs/1509.00622.

\bibitem{gabow}
Harold~N. Gabow and Robert~Endre Tarjan.
\newblock A linear-time algorithm for a special case of disjoint set union.
\newblock In {\em Proc. 15th STOC}, pages 246--251, 1983.

\bibitem{garelCPM}
Emmanuelle Garel.
\newblock Minimal separators of two words.
\newblock In {\em Proc. CPM 1993}, volume 684 of {\em Lecture Notes in Computer
  Science}, pages 35--53, 1993.

\bibitem{pawelLCE}
Pawel Gawrychowski, Tomasz Kociumaka, Wojciech Rytter, and Tomasz Walen.
\newblock Faster longest common extension queries in strings over general
  alphabets.
\newblock In {\em Proc. {CPM} 2016,}, volume~54 of {\em LIPIcs}, pages
  5:1--5:13, 2016.

\bibitem{mfcs2020}
Pawel Gawrychowski, Maria Kosche, Tore Koss, Florin Manea, and Stefan Siemer.
\newblock Efficiently testing simon's congruence.
\newblock {\em CoRR}, abs/2005.01112, 2020.
\newblock URL: \url{https://arxiv.org/abs/2005.01112}, \href
  {http://arxiv.org/abs/2005.01112} {\path{arXiv:2005.01112}}.

\bibitem{GawrychowskiRSS17}
Pawel Gawrychowski, Martin Lange, Narad Rampersad, Jeffrey~O. Shallit, and
  Marek Szykula.
\newblock Existential length universality.
\newblock In {\em Proc. {STACS} 2020}, volume 154 of {\em LIPIcs}, pages
  16:1--16:14, 2020.

\bibitem{GoecknerGHKKKS18}
Bennet Goeckner, Corbin Groothuis, Cyrus Hettle, Brian Kell, Pamela
  Kirkpatrick, Rachel Kirsch, and Ryan~W. Solava.
\newblock Universal partial words over non-binary alphabets.
\newblock {\em Theor. Comput. Sci}, 713:56--65, 2018.

\bibitem{HalfonSZ17}
Simon Halfon, Philippe Schnoebelen, and Georg Zetzsche.
\newblock Decidability, complexity, and expressiveness of first-order logic
  over the subword ordering.
\newblock In {\em Proc. {LICS} 2017}, pages 1--12, 2017.

\bibitem{6772729}
R.~W. {Hamming}.
\newblock Error detecting and error correcting codes.
\newblock {\em The Bell System Technical Journal}, 29(2):147--160, 1950.

\bibitem{DBLP:journals/ijfcs/HanKS13}
Yo{-}Sub Han, Sang{-}Ki Ko, and Kai Salomaa.
\newblock The edit-distance between a regular language and a context-free
  language.
\newblock {\em Int. J. Found. Comput. Sci.}, 24(7):1067--1082, 2013.

\bibitem{TCS::Hebrard1991}
Jean-Jacques Hebrard.
\newblock An algorithm for distinguishing efficiently bit-strings by their
  subsequences.
\newblock {\em Theoretical Computer Science}, 82(1):35--49, 22~May 1991.

\bibitem{Hirschberg75}
Daniel~S. Hirschberg.
\newblock A linear space algorithm for computing maximal common subsequences.
\newblock {\em Commun. {ACM}}, 18(6):341--343, 1975.

\bibitem{HolzerK11}
Markus Holzer and Martin Kutrib.
\newblock Descriptional and computational complexity of finite automata - {A}
  survey.
\newblock {\em Inf. Comput.}, 209(3):456--470, 2011.
\newblock \href {https://doi.org/10.1016/j.ic.2010.11.013}
  {\path{doi:10.1016/j.ic.2010.11.013}}.

\bibitem{union-find}
Hiroshi Imai and Takao Asano.
\newblock Dynamic segment intersection search with applications.
\newblock In {\em Proc. {FOCS} 1984}, pages 393--402, 1984.

\bibitem{DBLP:conf/icalp/JayaramS17}
Rajesh Jayaram and Barna Saha.
\newblock Approximating language edit distance beyond fast matrix
  multiplication: Ultralinear grammars are where parsing becomes hard!
\newblock In {\em Proc. {ICALP} 2017}, volume~80 of {\em LIPIcs}, pages
  19:1--19:15, 2017.

\bibitem{KarandikarKS15}
Prateek Karandikar, Manfred Kufleitner, and Philippe Schnoebelen.
\newblock On the index of {S}imon's congruence for piecewise testability.
\newblock {\em Inf. Process. Lett.}, 115(4):515--519, 2015.
\newblock \href {https://doi.org/10.1016/j.ipl.2014.11.008}
  {\path{doi:10.1016/j.ipl.2014.11.008}}.

\bibitem{CSLKarandikarS}
Prateek Karandikar and Philippe Schnoebelen.
\newblock The height of piecewise-testable languages with applications in
  logical complexity.
\newblock In {\em Proc. {CSL} 2016}, volume~62 of {\em LIPIcs}, pages
  37:1--37:22, 2016.

\bibitem{journals/lmcs/KarandikarS19}
Prateek Karandikar and Philippe Schnoebelen.
\newblock The height of piecewise-testable languages and the complexity of the
  logic of subwords.
\newblock {\em Logical Methods in Computer Science}, 15(2), 2019.

\bibitem{dimaRuns}
Dmitry Kosolobov.
\newblock Computing runs on a general alphabet.
\newblock {\em Inf. Process. Lett.}, 116(3):241--244, 2016.

\bibitem{dimaCF}
Dmitry Kosolobov.
\newblock Finding the leftmost critical factorization on unordered alphabet.
\newblock {\em Theor. Comput. Sci.}, 636:56--65, 2016.

\bibitem{KrotzschMT17}
Markus Kr{\"{o}}tzsch, Tom{\'{a}}s Masopust, and Micha{\"{e}}l Thomazo.
\newblock Complexity of universality and related problems for partially ordered
  {NFA}s.
\newblock {\em Inf. Comput.}, 255:177--192, 2017.
\newblock \href {https://doi.org/10.1016/j.ic.2017.06.004}
  {\path{doi:10.1016/j.ic.2017.06.004}}.

\bibitem{Kuske20}
Dietrich Kuske.
\newblock The subtrace order and counting first-order logic.
\newblock In {\em Proc. {CSR} 2020}, volume 12159 of {\em Lecture Notes in
  Computer Science}, pages 289--302, 2020.

\bibitem{KuskeZ19}
Dietrich Kuske and Georg Zetzsche.
\newblock Languages ordered by the subword order.
\newblock In {\em Proc. {FOSSACS} 2019}, volume 11425 of {\em Lecture Notes in
  Computer Science}, pages 348--364, 2019.

\bibitem{Rigo19}
Marie Lejeune, Julien Leroy, and Michel Rigo.
\newblock Computing the k-binomial complexity of the {T}hue-{M}orse word.
\newblock In {\em Proc. {DLT} 2019}, volume 11647 of {\em Lecture Notes in
  Computer Science}, pages 278--291, 2019.

\bibitem{LeroyRS17a}
Julien Leroy, Michel Rigo, and Manon Stipulanti.
\newblock Generalized {P}ascal triangle for binomial coefficients of words.
\newblock {\em Electron. J. Combin.}, 24(1.44):36 pp., 2017.

\bibitem{EditDistance}
Vladimir~I. Levenshtein.
\newblock Binary codes capable of correcting deletions, insertions, and
  reversals.
\newblock {\em Soviet Physics Doklady}, 10(8):707--710, 1966.

\bibitem{Maier:1978}
David Maier.
\newblock The complexity of some problems on subsequences and supersequences.
\newblock {\em J. ACM}, 25(2):322--336, April 1978.
\newblock URL: \url{http://doi.acm.org/10.1145/322063.322075}, \href
  {https://doi.org/10.1145/322063.322075} {\path{doi:10.1145/322063.322075}}.

\bibitem{martin1934}
Monroe~H. Martin.
\newblock A problem in arrangements.
\newblock {\em Bull. Amer. Math. Soc.}, 40(12):859--864, 12 1934.
\newblock URL: \url{https://projecteuclid.org:443/euclid.bams/1183497876}.

\bibitem{DBLP:journals/jcss/MasekP80}
William~J. Masek and Mike Paterson.
\newblock A faster algorithm computing string edit distances.
\newblock {\em J. Comput. Syst. Sci.}, 20(1):18--31, 1980.

\bibitem{Mat04}
Alexandru Mateescu, Arto Salomaa, and Sheng Yu.
\newblock Subword histories and {P}arikh matrices.
\newblock {\em Journal of Computer and System Sciences}, 68(1):1--21, 2004.

\bibitem{Pin2004}
Jean{-}Eric Pin.
\newblock The consequences of imre simon's work in the theory of automata,
  languages, and semigroups.
\newblock In {\em Proc. {LATIN} 2004}, volume 2976 of {\em Lecture Notes in
  Computer Science}, page~5, 2004.

\bibitem{Pin2019}
Jean{-}Eric Pin.
\newblock The influence of {Imre Simon's} work in the theory of automata,
  languages and semigroups.
\newblock {\em Semigroup Forum}, 98:1--8, 2019.

\bibitem{Rampersad:2012}
Narad Rampersad, Jeffrey Shallit, and Zhi Xu.
\newblock The computational complexity of universality problems for prefixes,
  suffixes, factors, and subwords of regular languages.
\newblock {\em Fundam. Inf.}, 116(1-4):223--236, January 2012.
\newblock URL: \url{http://dl.acm.org/citation.cfm?id=2385073.2385090}.

\bibitem{RigoS15}
Michel Rigo and Pavel Salimov.
\newblock Another generalization of abelian equivalence: Binomial complexity of
  infinite words.
\newblock {\em Theor. Comput. Sci.}, 601:47--57, 2015.

\bibitem{Loth97}
Jacques Sakarovitch and Imre Simon.
\newblock Subwords.
\newblock In M.~Lothaire, editor, {\em Combinatorics on Words}, chapter~6,
  pages 105--142. Cambridge University Press, 1997.

\bibitem{Salomaa05}
Arto Salomaa.
\newblock Connections between subwords and certain matrix mappings.
\newblock {\em Theoretical Computer Science}, 340(2):188--203, 2005.

\bibitem{sankoff}
David Sankoff and Joseph Kruskal.
\newblock {\em Time Warps, String Edits, and Macromolecules The Theory and
  Practice of Sequence Comparison}.
\newblock Cambridge University Press, 2000 (reprinted).
\newblock originally published in 1983.

\bibitem{Seki12}
Shinnosuke Seki.
\newblock Absoluteness of subword inequality is undecidable.
\newblock {\em Theor. Comput. Sci.}, 418:116--120, 2012.
\newblock \href {https://doi.org/10.1016/j.tcs.2011.10.017}
  {\path{doi:10.1016/j.tcs.2011.10.017}}.

\bibitem{simonPhD}
Imre Simon.
\newblock {\em Hierarchies of events with dot-depth one - Ph.D. thesis}.
\newblock University of Waterloo, 1972.

\bibitem{Simon72}
Imre Simon.
\newblock Piecewise testable events.
\newblock In {\em Autom.\ Theor.\ Form.\ Lang., 2nd GI Conf.}, volume~33 of
  {\em LNCS}, pages 214--222, 1975.

\bibitem{SimonWords}
Imre Simon.
\newblock Words distinguished by their subwords (extended abstract).
\newblock In {\em Proc. {WORDS} 2003}, volume~27 of {\em TUCS General
  Publication}, pages 6--13, 2003.

\bibitem{DBLP:conf/wia/Tronicek02}
Zdenek Tron{\'{\i}}cek.
\newblock Common subsequence automaton.
\newblock In {\em Proc. {CIAA} 2002 (Revised Papers)}, volume 2608 of {\em
  Lecture Notes in Computer Science}, pages 270--275, 2002.

\bibitem{Wagner:1974}
Robert~A. Wagner and Michael~J. Fischer.
\newblock The string-to-string correction problem.
\newblock {\em J. ACM}, 21(1):168--173, January 1974.
\newblock URL: \url{http://doi.acm.org/10.1145/321796.321811}, \href
  {https://doi.org/10.1145/321796.321811} {\path{doi:10.1145/321796.321811}}.

\bibitem{Zetzsche16}
Georg Zetzsche.
\newblock The complexity of downward closure comparisons.
\newblock In {\em Proc. {ICALP} 2016}, volume~55 of {\em LIPIcs}, pages
  123:1--123:14, 2016.

\end{thebibliography}
